%% file: main.tex
\newif\ifsubmission
\newif\iffull\fulltrue
\newif\ifndss\ndsstrue
\ifsubmission
\documentclass[sigconf, anonymous]{acmart}
\else
\documentclass[sigconf]{acmart}
\fi

\setcopyright{none} 
\acmConference[]{}{}{}

\usepackage{framed}
\usepackage{minibox}
\usepackage{amsmath,amssymb,amsthm,amsfonts}
\usepackage{graphicx}
\usepackage{xcolor}
\usepackage[T1]{fontenc}
\usepackage[ruled,algosection,noend,noline,linesnumbered]{algorithm2e}
\usepackage{multirow,array}
\usepackage{subcaption}
\usepackage[font=small]{caption}
\usepackage{enumitem}
\usepackage[all]{xy}
\usepackage{booktabs}
\usepackage{multirow}
\usepackage{url}
\usepackage{xspace}
\usepackage[group-separator={,}]{siunitx}
\usepackage{epigraph}
\usepackage{balance}

\newtheorem{theorem}{Theorem}[section]
\newtheorem{claim}[theorem]{Claim}

\newenvironment{denseitemize}{
	\begin{itemize}
		\setlength{\itemsep}{1pt}
		\setlength{\parskip}{0pt}
		\setlength{\parsep}{0pt}
	}{\end{itemize}}

\newcommand\bi{\begin{denseitemize}}
	\newcommand\ei{\end{denseitemize}}
\newcommand\ben{\begin{enumerate}}
	\newcommand\een{\end{enumerate}}
\newcommand{\trm}[1]{\textrm{#1}}

\newcommand*\dash{\ifvmode\quitvmode\else\unskip\kern.16667em\fi---%
	\hskip.16667em\relax}
\newcommand{\sysname}{Caucus\xspace}
\newcommand{\betting}{Fant\^omette\xspace}
\newcommand\sarahm[1]{}
\newcommand\saraha[1]{}
\newcommand\paddy[1]{}

\input{macros}

\begin{document}
\title{Betting on Blockchain Consensus with \betting}

\input{abstract}

\iffull{
\author{Sarah Azouvi, Patrick McCorry, and Sarah Meiklejohn}
\affiliation{University College London}
\email{{sarah.azouvi.13,p.mccorry,s.meiklejohn}@ucl.ac.uk}
\fi

\maketitle


\input{intro}

\input{related}
\input{defns}
\input{model}
\input{leaderelection}
\input{protocol}

\input{simulation}

\input{security}
\input{conclusion}
\section*{Acknowledgements}

All authors are supported in part by EPSRC Grant EP/N028104/1.

{
\balance
{\footnotesize
\def\shortbib{1}
\bibliographystyle{abbrv}
\begin{flushleft}
\bibliography{abbrev2,crypto,misc,ref,spw}
\end{flushleft}
}
}

\appendix

\ifsubmission\input{app-defns}\fi
\ifsubmission\input{app-caucus}\fi
\ifsubmission\input{app-details}\fi	

\end{document}

%% file: macros.tex
\newcommand{\pun}{\mathsf{pun}}
\newcommand{\bigpun}{\mathsf{bigpun}}
\newcommand{\rwd}{\mathsf{rwd}}

\newcommand{\rank}{\mathsf{rk}}
\newcommand{\eligible}{\mathsf{Eligible}}

\newcommand{\tree}{\mathsf{G}}
\newcommand{\blockdag}{\mathsf{G}}
\newcommand{\trees}{\mathcal{G}}

\newcommand{\directfuture}{\mathsf{DirectFuture}}
\newcommand{\past}{\mathsf{Past}}

\newcommand{\coalition}{\mathcal{C}}

\newcommand{\fcr}{\mathsf{FCR}}
\newcommand{\winner}{\trm{winner}}
\newcommand{\neutral}{\trm{neutral}}
\newcommand{\loser}{\trm{loser}}
\newcommand{\lab}{\mathsf{Label}}
\newcommand{\maps}{\mathcal{M}}
\newcommand{\map}{\mathsf{M}}
\newcommand\bcpc{\mathsf{BCPD}}

\newcommand\treepub{\tree_{\mathsf{pub}}}

\newcommand\precede\preccurlyeq
\newcommand\maxweight{\mathsf{CW}}
\newcommand{\byzfrac}{f_B}
\newcommand{\ratfrac}{f_R}
\newcommand{\ncoal}{n_\coalition}
\newcommand\altfrac{f_A}
\newcommand\coalfrac{f_\coalition}
\newcommand{\leaves}{\blocks_\leaf}
\newcommand{\dagleaves}{\mathsf{Leaves}}
\newcommand{\leaf}{\mathsf{leaf}}
\newcommand{\blue}{\mathsf{blue}}

\newcommand{\ancestors}{\mathsf{Ancestors}}
\newcommand{\bone}{\block_1}
\newcommand{\btwo}{\block_2}
\newcommand{\gen}{\mathsf{Gen}}
\newcommand{\provepef}{\mathsf{Prove}}
\newcommand\eligibilitytest{VRF }
\newcommand\double{\mathsf{Double}}
\newcommand{\pef}{\trm{G}}

\newcommand{\pr}{\textrm{Pr}}

\newcommand{\A}{\mathcal{A}}
\newcommand{\N}{\mathbb{N}}

\newcommand{\secp}{\lambda}
\newcommand{\usecp}{1^{\secp}}

\newcommand{\stateinfo}{\mathsf{state}}

\newcommand{\randpick}{%
  \mathrel{\vbox{\offinterlineskip\ialign{%
    \hfil##\hfil\cr
    $\scriptscriptstyle\$$\cr
    \noalign{\kern-.1ex}
    $\leftarrow$\cr
}}}}

\newcommand{\bet}{\mathsf{Bet}}

\newcommand{\inputs}{\mathsf{inputs}}
\newcommand{\outputs}{\mathsf{outputs}}

\newcommand{\pk}{\ensuremath{pk}}
\newcommand{\sk}{{\ensuremath{sk}}}

\newcommand{\blocks}{\mathcal{B}}

\newcommand{\block}{\mathsf{B}}
\newcommand{\chain}{\mathsf{C}}

\newcommand{\round}{\mathsf{rnd}}

\newcommand{\blockprev}{\block_{\mathsf{prev}}}
\newcommand{\anticone}{\mathsf{Anticone}}
\newcommand{\sender}{\mathsf{snd}}

\newcommand{\tx}{\mathsf{tx}}

\newcommand{\verifyblock}{\mathsf{VerifyBlock}}

\newcommand{\setup}{\mathsf{Setup}}
\newcommand{\verify}{\mathsf{Verify}}

\newcommand{\coord}{\mathtt{Update}}

\newcommand{\privatestate}{\stateinfo_{\mathsf{priv}}}

\newcommand{\pubstate}{\stateinfo_{\mathsf{pub}}}

\newcommand{\target}{\mathsf{target}}
\newcommand{\txset}{\mathsf{txset}}

\newcommand{\participants}{\mathcal{P}}

\newcommand{\afrac}{f_\A}

\newcommand{\globalrandao}{R}

\newcommand{\reveal}{\mathsf{Reveal}}
\newcommand{\commit}{\mathsf{Commit}}
\newcommand{\barrier}{\mathsf{barrier}}

\newcommand{\participantsnumber}{n}

\newcommand{\hashmax}{H_{\mathsf{max}}}
\newcommand{\txdep}{\tx_{\mathsf{com}}}
\newcommand{\txrev}{\tx_{\mathsf{rev}}}

\newcommand{\commitlist}{\mathtt{c}}

\newcommand{\deposit}{\mathsf{Commit}}

%% file: abstract.tex
\begin{abstract}
Blockchain-based consensus protocols present the opportunity to develop
new protocols, due to their novel requirements of open participation and
explicit incentivization of participants.  To address the first requirement,
it is necessary to consider the \emph{leader election} inherent in consensus
protocols, which can be difficult to scale to a large and untrusted set of
participants.  To address the second, it is important to consider ways to
provide incentivization without relying on the resource-intensive
proofs-of-work used in Bitcoin. In this paper, we propose a secure leader election
protocol, \sysname; we next fit this
protocol into a broader blockchain-based consensus protocol, \betting, that
provides game-theoretic guarantees in addition to traditional
blockchain security properties.
\betting is the first proof-of-stake protocol
to give formal game-theoretic proofs of security in the presence
of non-rational players.

\end{abstract}

%% file: intro.tex
\section{Introduction}

One of the central components of any distributed system is a \emph{consensus
protocol}, by which the system's participants can agree on its current
state and use that information to take various actions.  Consensus protocols
have been studied for decades in the distributed systems literature, and
classical protocols such as Paxos~\cite{paxos-made-simple} and PBFT~\cite{pbft} have 
emerged as ``gold standards'' of sorts, in terms of their ability to guarantee
the crucial properties of safety and liveness even in the face of faulty or
malicious nodes.\\
\indent The setting of \emph{blockchains} has renewed interest in consensus protocols, 
due largely to two crucial new requirements:
scalability and incentivization.  First, classical consensus
protocols were designed for a closed and relatively small set of participants,
whereas in open (or ``permissionless'') blockchains the goal is to enable
anyone to join.  This requires the design of new consensus protocols that can
both scale to handle a far greater number of participants, and also ones that
can address the question of Sybil attacks~\cite{sybil}, due to the fact that
participants may no longer be well identified.\\
\indent
At heart, one of the biggest obstacles in scaling classical consensus
protocols is in scaling their underlying \emph{leader election} protocol, 
in which one participant or subset of participants is chosen to lead the
decisions around what information should get added to the ledger for a single 
round (or set period of time).  To elect a leader, participants must coordinate 
amongst themselves by exchanging several rounds of messages.  If a Sybil is
elected leader, this can result in the adversary gaining control over the
ledger, and if there are too many participants then the exchange of messages
needed to carry out the election may become prohibitively expensive.\\
\indent
Many recent proposals for blockchain-based consensus protocols focus on 
solving this first requirement by presenting more scalable leader election
protocols~\cite{praos,algorand,snow,thunderella}. 
\\
\indent
The second novel requirement of blockchains is the explicit economic
incentivization on behalf of participants.  In contrast to classical 
consensus protocols, where it is simply assumed that some set of nodes is 
interested in coming to consensus, in Bitcoin this
incentivization is created through the use of block rewards and transaction
fees.  Again, several recent proposals have worked to address this question of
incentives~\cite{ouroboros,fruitchains,spacemint,solidus,casper}.  Typically, 
however, the
analysis of these proposals has focused on proving that following the protocol
is a Nash equilibrium (NE), which captures the case of rational players but not
ones that are Byzantine (i.e., fully malicious).\\
\indent
Indeed, despite these advances, one could argue that the only consensus 
protocol to fully address both requirements is still the one underlying 
Bitcoin, which is typically known as
Nakamoto or proof-of-work (PoW)-based consensus.  This has in fact been proved
secure~\cite{Garay2015,blockchain-asynchronous}, but provides Sybil 
resistance only by requiring enormous
amounts of computational power to be expended.  As such, it has also been
heavily criticized for the large amount of electricity that it uses.
Furthermore, it has other subtle limitations; e.g., it does not achieve any
notion of fully deterministic finality and some attacks have been found on its
incentive scheme~\cite{optimal-selfish-mining,bitcoin-without-block-reward}.

\paragraph*{Our contributions}

In this paper, we propose \betting, a new blockchain-based consensus protocol
that fully incorporates an incentive design to prove security properties in a settings
that considers both rational and Byzantine adversaries.
Our
initial observation is that the PoW-based setting contains an implicit
investment on the part of the miners, in the form of the costs of hardware and
electricity.  In moving away from PoW, this implicit investment no longer
exists, giving rise to new potential attacks due to the fact that
creating blocks is now costless.  It is thus necessary to compensate by adding 
explicit punishments into the protocol for participants who misbehave. This 
is difficult to do in a regular blockchain setting. 
In particular, 
blockchains do not reveal information about which other blocks miners may have 
been aware of at the time they produced their block, so they cannot be 
punished for making ``wrong'' choices.  We thus move to the setting of 
\emph{blockDAGs}~\cite{spectre,phantom}, which 
induce a more complex fork-choice rule and expose more of the decision-making 
process of participants.
Within this model, we are able to leverage the requirement that players must
place \emph{security deposits} in advance of participating in the consensus
protocol to achieve two things.  First, we can implement punishments by taking
away some of the security deposit, and thus incentivize rational 
players to follow the protocol.  Second, because this allows the players to 
be identified, we can provide a decentralized \emph{checkpointing} system,
which in turn allows us to achieve a notion of finality.\\
\indent
Along the way, we present in Section~\ref{sec:construction} a 
leader election protocol, \sysname, that is specifically designed for open
blockchains, and that we prove secure in a model presented in
Section~\ref{sec:model}.  
We then use \sysname as a component in the broader \betting consensus
protocol, which we present in Section~\ref{sec:protocol} and argue for the 
security of in Section~\ref{sec:sim}.  Here we rely on \sysname to address the
first requirement of scaling in blockchain-based consensus protocols, so can 
focus almost entirely on the second requirement of incentivization.  
\\
\indent
In summary, we make the following concrete contributions:
(1) we present the design of a leader election 
protocol, \sysname.  In addition to provably satisfying more traditional 
notions of security, \sysname has several ``bonus'' properties; e.g., it 
ensures that leaders are revealed only when they take action, which prevents
them from being subject to the DoS attacks possible when their eligibility is 
revealed ahead of time.
(2) we present the design and simulation of a full blockchain-based
consensus protocol, \betting, that provides a scheme for incentivization that
is robust against both rational and fully adaptive Byzantine adversaries.  
\betting is compatible with proof-of-stake (PoS), but could also be
used for other ``proof-of-X'' settings with an appropriate leader election
protocol.

%% file: related.tex
\section{Related}\label{sec:related}

The work that most closely resembles ours is the cryptographic literature on
proof-of-stake (PoS). We evaluate and compare each protocol along the two
requirements outlined in the introduction of scalability and incentivization.
Other non-academic work proposes PoS 
solutions~\cite{ppcoin,neucoin,tendermint}, which is related to
\betting in terms of the recurrent theme of punishment in the case of 
misbehavior.\\
\indent
In Ouroboros~\cite{ouroboros}, the honest strategy is a $\delta$-Nash 
equilibrium, which addresses the question of incentives.  The leader election,
however, is based on a coin-tossing scheme that requires a relatively large
overhead.  This is addressed in Ouroboros Praos~\cite{praos}, which utilizes
the same incentive structure but better addresses the question of scalability
via a more efficient leader election protocol (requiring, as we do in
\sysname, only one broadcast message to prove eligibility).  A recent
improvement, Ouroboros Genesis~\cite{ouroboros-genesis}, allows for dynamic
availability; i.e., allows offline parties to safely bootstrap the blockchain
when they come back online.\\
\indent
A comparable protocol is Algorand~\cite{algorand}, which proposes
a scalable Byzantine agreement protocol and the use of a ``cryptographic
sortition'' leader election protocol, which resembles \sysname.  They do not
address the question of incentives. \\
\indent
In Snow White~\cite{snow}, the incentive structure is based on
that of Fruitchains~\cite{fruitchains}, where honest mining is a 
NE resilient to coalitions.  The incentive structure of
Thunderella~\cite{thunderella} is also based on Fruitchains, but is PoW-based.\\
\indent
Casper~\cite{casper-econ} is still work in progress, so it is difficult to say
how well it addresses scalability.  On the topic of incentivization, it
proposes that following that protocol should be a Nash equilibrium
and that an attacker should lose more in an attack than the victims
of the attack.
A closely related recent paper is Hot-Stuff~\cite{hotstuff}, which 
proposes a PBFT-style consensus protocol that operates in an asynchronous
model.  Again, this paper does not address incentivization.
\\
\indent
Among other types of consensus protocols, there are several that do closely 
consider the topic of incentivization.  The first version of 
Solidus~\cite{solidus}, which is a consensus protocol based on PoW, provides 
a $(k,t)$-robust equilibrium for any 
$k+t<f$, although they leave a rigorous proof of this for future work.
SpaceMint~\cite{spacemint} is a cryptocurrency based on proof-of-space, and
they prove that following the protocol is a Nash equilibrium.
In terms of protocols based on PoW, SPECTRE~\cite{spectre} introduced the 
notion of a blockDAG, which was further refined by PHANTOM~\cite{phantom}.
Our \betting protocol is inspired by PHANTOM (although translated to a non-PoW
setting), and in particular we leverage the notion of connectivity of blocks 
induced by blockDAGs.
Avalanche~\cite{avalanche}, a recent proposal also relying on blockDAG, but in the
context of PoS, recently appeared.
They propose a PBFT-style
consensus protocol that achieves $O(kn)$ for some $k\ll n$, and is thus more expensive than
Caucus that requires a single message broadcast.
They do not consider the question of incentives.
\\
\indent
In terms of economic analyses, Kroll et al.~\cite{Kroll_theeconomics}
studied the economics of PoW and showed that following the protocol is a NE.
Badertscher et al~\cite{butwhy} analyze Bitcoin in the Rational
Protocol Design setting.
The ``selfish mining'' series of 
attacks~\cite{eyal2015minersDilemma,verifierdilemma,optimal-selfish-mining} 
show that the incentive structure of Nakamoto consensus is vulnerable.
Carlsten et al.~\cite{bitcoin-without-block-reward} also consider an attack on 
the incentives in Bitcoin, in the case in which there are no block rewards.
Recently, Ga{\v z}i et al.~\cite{stake-bleeding} proposed an attack on
PoS protocols that we address here.
\\
\indent
Beyond the blockchain setting, Halpern~\cite{beyond} presents 
new perspectives from game theory that go beyond Nash equilibria in terms of
analyzing incentive structures.  One of this concept comes from Abraham et 
al.~\cite{distributed-computing-game-theory}, which we use in
Section~\ref{sec:game-theory}.
\\
\indent
To summarize, most existing proposals for non-PoW blockchain consensus
protocols provide a basic game-theoretic analysis, using techniques like Nash 
equilibria, but do not consider more advanced economic analyses that tolerate 
coalitions or the presence of Byzantine adversaries.  \betting is thus
the first one to place incentivization at the core of its security.  In terms
of scalability, our protocol is again different as it is
not based on a BFT-style algorithm, but is rather inspired by Nakamoto 
consensus (leveraging its economic element).

%% file: defns.tex
\section{Background Definitions and Notation}\label{sec:defns}

In this section, we present the underlying definitions we rely on in the rest
of the paper. \ifndss \iffull{We begin (Sections~\ref{sec:notation}-\ref{sec:defns-cointoss}) 
	with the cryptographic notation and primitives
	necessary for our leader election protocol, \sysname.}\else{Due to space
	constraints, the cryptographic notation and primitives necessary for our
	leader election protocol, \sysname, can be found in
	Appendix~\ref{sec:defns-app}.}\fi  \fi
\iffull\input{prelim}\fi
\newcommand\hashdefs{
}
\iffull\hashdefs\fi
Verifiable Random Functions (VRF), first introduced by Micali et al.~\cite{verifiablerandomfunctions},
generate a pseudo-random number in a publicly verifiable way.
Formally, a VRF is defined as follows:
\begin{definition}
A VRF consists in three  polynomial-time algorithm $(\gen,\provepef,\verify)$
that works as follows: 
(1) $\gen(\usecp)$ outputs a key pair $(\pk,\sk)$;
(2) $\provepef_{\sk}(x)$ outputs a pair $(y=\pef_{sk}(x),p=p_{\sk}(x))$;
(3) $\verify(x,y,p)$ verifies that
$y=\pef_{sk}(x)$ using $p$.
An VRF satisfies correctness if:
\begin{itemize}
\item if $(y,p)=\provepef_{\sk}(x)$ then $\verify(x,y,p)=1$ 
\item for every $(\sk,x)$ there is a unique $y$ such that $\verify(x,y,p_{\sk}(x))=1$
\item it verifies pseudo-randomness: for any PPT algorithm $A=(A_E,A_J)$:
\small{
\[ \mathbb{P} \left[ \begin{array}{c|c}
 & (\pk,\sk)\gets\gen(\usecp); \\ b=b' & (x,A_{st})\gets A_E^{\provepef(.)}(\pk); \\  & y_0=\pef_{\sk}(x); y_1\gets\{0,1\}^{len(\pef)};\\
 & b\gets\{0,1\}; b'\gets A_J^{\provepef(.)}(y_b,A_{st}) \end{array} \right] \le \frac{1}{2} +negl(k)
\]}
\end{itemize}
\end{definition}

\newcommand\beacondefs{
	\subsection{Coin tossing and random beacons}\label{sec:defns-cointoss}
	
	Coin tossing is closely related to leader
	election~\cite{collectivecoinflipping}, and allows two or more parties to
	agree on a single or many random bits~\cite{blum-cointossing,collectivecoinflipping,EPRINT:Popov16};
	i.e., to output a value $R$ that is statistically close to
	random.
	A coin-tossing protocol must satisfy \emph{liveness}, \emph{unpredictability},
	and \emph{unbiasability}~\cite{randhound}, where we define these (in keeping with our
	definitions for leader election in Section~\ref{sec:model}) as follows:
	
	\begin{definition}\label{def:coin-liveness}
		Let $\afrac$ be the fraction of participants controlled by an adversary $\A$.
		Then a coin-tossing protocol satisfies \emph{$\afrac$-liveness} if it is still
		possible to agree on a random value $R$ even in the face of such an $\A$.
	\end{definition}
	
	\begin{definition}\label{def:coin-unpredictability}
		A coin-tossing protocol satisfies \emph{unpredictability} if, prior to some step
		$\barrier$ in the protocol, no PT adversary can produce better than a random
		guess at the value of $R$.
	\end{definition}
	
	\begin{definition}\label{def:coin-unbiasability}
		A coin-tossing protocol is \emph{$\afrac$-unbiasable} if for all PT
		adversaries $\A$ controlling an $\afrac$ fraction of participants, the
		output $R$ is still statistically close to a uniformly distributed random
		string.
	\end{definition}
	A concept related to coin tossing is \emph{random beacons}.  These
	were first introduced by Rabin~\cite{rab81} as a service for ``emitting at
	regularly spaced time intervals, randomly chosen integers''.  To extend the
	above definitions to random beacons, as inspired by~\cite{EPRINT:BonClaGol15}, we require that
	the properties of $\afrac$-liveness and $\afrac$-unbiasability apply
	for each iteration of the beacon, or \emph{round}.  We also require that the
	$\barrier$ in the unpredictability definition is at least the
	beginning of each round.  
}
\iffull\beacondefs\fi

\subsection{Blockchains}\label{sec:defns-blockchain}

Distributed ledgers, or blockchains, have become increasingly popular ever
since Bitcoin was first proposed by Satoshi Nakamoto in
2008~\cite{satoshi-bitcoin}.
Briefly, individual Bitcoin users wishing to pay other users broadcast
\emph{transactions} to a global peer-to-peer network, and the peers
responsible for participating in Bitcoin's consensus protocol (i.e., for
deciding on a canonical ordering of transactions) are known as \emph{miners}.
Miners form \emph{blocks}, which contain (among
other things that we ignore for ease of exposition) the transactions collected
since the
previous block was mined, which we denote $\txset$, a pointer to the previous
block hash $h_{\mathsf{prev}}$, and a \emph{proof-of-work} (PoW).  This PoW
is the solution to a computational puzzle.
If future miners choose to include this
block by incorporating its hash into their own PoW then it becomes part of the 
global \emph{blockchain}.
A \emph{fork} is created if two miners find a block at 
the same time.  Bitcoin follows the \emph{longest chain} rule, meaning that 
whichever fork creates the longest chain is the one that is considered valid.  
%


\subsubsection*{BlockDAGs}

As the name suggests, a blockchain is a chain of blocks, with
each block referring to only one previous block.  In contrast, a
\emph{blockDAG}~\cite{spectre,phantom}, is a directed acyclic graph (DAG) of 
blocks.  In our paper (which slightly adapts the original definitions), every
block still specifies a single parent block, but can also \emph{reference} 
other recent blocks of which it is aware.  These are referred to as 
\emph{leaf blocks}, as they are leaves in the tree (also sometimes
referred as the tips of the chain), and are denoted $\dagleaves(\tree)$.   
Blocks thus have the form $\block=(\blockprev,\leaves,\pi,\txset)$ where
$\blockprev$ is the parent block, $\leaves$ is a list of previous leaf blocks, 
$\pi$ is a proof of some type of eligibility (e.g., a PoW), and $\txset$ is 
the transactions contained within the block.  
We denote by $\block.\sender$ the participant that created block $\block$.
%
In addition, we define the following notation: 


\begin{itemize}[leftmargin=0.2cm]

\item $\tree$ denotes a DAG, $\tree^i$ the DAG according to a 
participant $i$, and $\trees$ the space of all possible DAGs.

\item In a block, $\blockprev$ is a direct \emph{parent} of $\block$, and 
$\ancestors(\block)$ denotes the set of all blocks that are parents of 
$\block$, either directly or indirectly.

\item A \emph{chain} is a set of blocks $\mathcal{M}$ such that there exists 
one block $\block$ for which $\mathcal{M}=\ancestors(\block)$.

\item $\past(\block)$ denotes the subDAG consisting of all the blocks that 
$\block$ references directly or indirectly.

\item $\directfuture(\block)$ denotes the set of blocks that directly 
reference $\block$ (i.e., that include it in $\leaves$). 

\item $\anticone(\block)$ denotes the set of blocks $\block'$ such that 
$\block'\notin\past(\block)$ and $\block \notin \past(\block')$.

\item $d$ denotes the \emph{distance} between two blocks in the DAG.

\item The \emph{biggest common prefix DAG} ($\bcpc$) is the biggest subDAG 
that is agreed upon by more than 50\% of the players.
\end{itemize}

\begin{figure}[t]
\centering
	\includegraphics[width=0.7\linewidth]{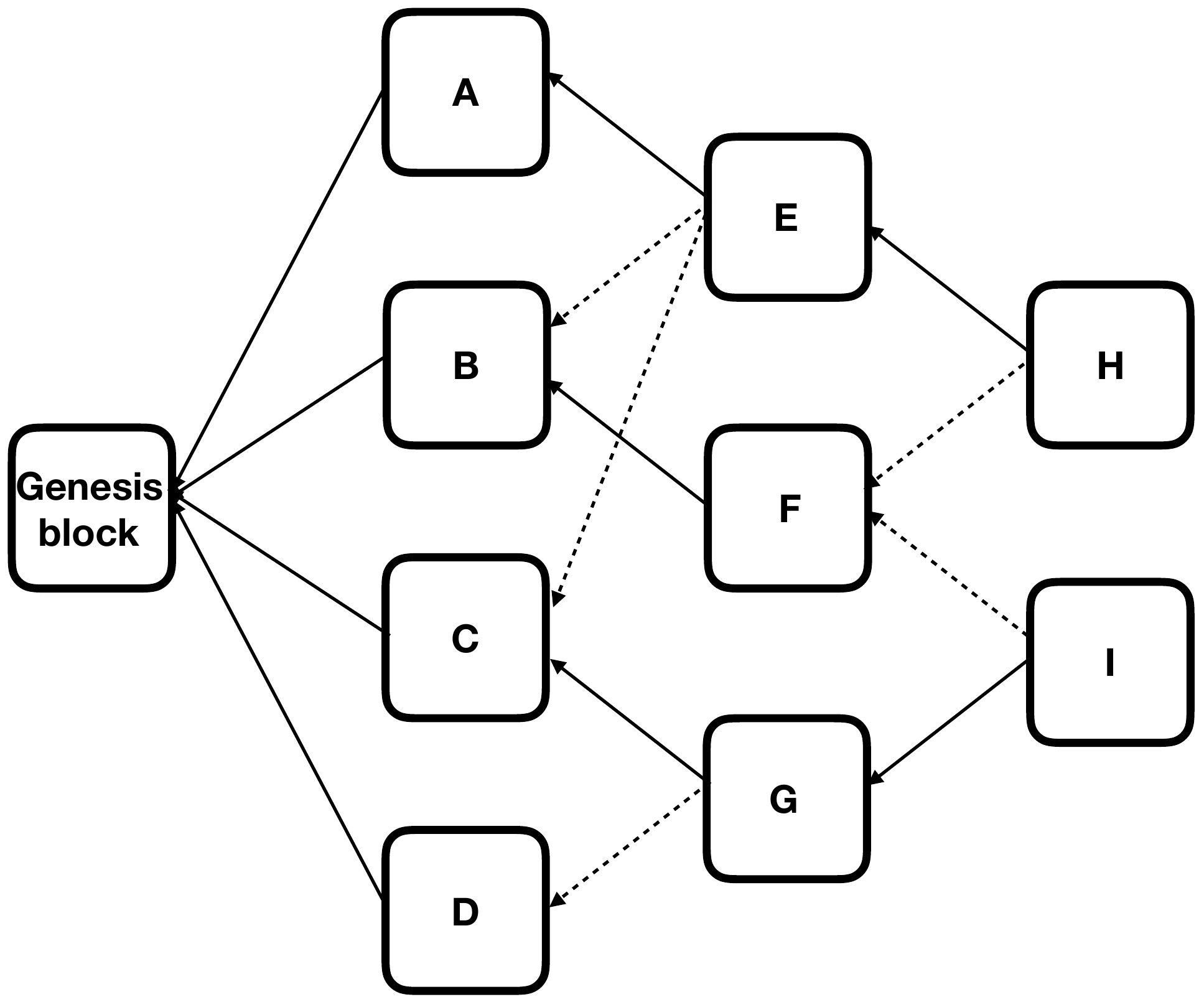}
	\caption{Example of a blockDAG.  A full arrow indicates a bet ($\blockprev$) 
		and a dashed arrow indicates a reference ($\leaves$).}
	\label{fig:blockdag}
\end{figure}

For example, if we consider the blockDAG in Figure~\ref{fig:blockdag},
we have that $\ancestors(H)=\{E,A,\trm{genesis}\}$
and $\{H,E,A,\trm{genesis}\}$ forms a chain.  We also have that 
$\directfuture(F)=\{H,I\}$, $\past(H)=\{E,F,A,B,C,\trm{genesis}\}$, 
$\anticone(E)=\{D,G,I\}$, and $d(A,H)=2$.


\subsubsection*{Proof-of-stake}\label{sec:defns-pos}

By its nature, PoW consumes a lot of energy.  Thus, some alternative consensus
protocols have been proposed that are more cost-effective;
arguably the most popular of these is called \emph{proof-of-stake}
(PoS)~\cite{pos-forum,ppcoin,casper}. If
we consider PoW to be a leader election protocol in which the leader (i.e.,
the miner with the valid block) is selected in proportion to their amount of
computational power, then PoS can be seen as a leader election protocol in
which the leader (i.e., the participant who is \emph{eligible} to propose a
new block) is selected in proportion to some ``stake'' they have in the
system (e.g. the amount of coins they have).  
\\
\indent
As security no longer stems from the fact that it is expensive to create a
block, PoS poses several technical challenges~\cite{interactivepos}.  The
main three are as follows: first, the \emph{nothing at stake} problem says
that miners have no reason to not mine on top of every chain, since mining is
costless, so it is more difficult to reach consensus.  This is an issue of
incentives that, as we present in Section~\ref{sec:protocol}, \betting 
aims to solve. We do so by giving rewards to players that follow the protocol
and punishing players who mine on concurrent chains.
\\
\indent
Second, PoS allows for \emph{grinding attacks}~\cite{interactivepos}, in 
which once a miner is
elected leader they privately iterate through many valid blocks (again,
because mining is costless) in an attempt to find one that may give them an
unfair advantage in the future (e.g., make them more likely to be elected
leader).  This is an issue that we encounter in Section~\ref{sec:construction}
and address using an unbiasable source of randomness.\\
\indent
Finally, in a \emph{long-range attack}, an attacker may bribe miners into
selling their old private keys, which would allow them to re-write the entire
history of the blockchain.  Such a ``stake-bleeding'' 
attack~\cite{stake-bleeding} can be launched
against PoS protocols that adopt a standard longest chain rule and do not
have some form of checkpointing.  We solve this problem in \betting by adding 
a notion of finality, in the form of decentralized checkpointing. 

\subsection{Game-theoretic definitions}\label{sec:game-theory}

In game-theoretic terms, we consider the consensus protocol as
a game in infinite extensive-form with imperfect
information~\cite{econ-notes}, where the utilities
depend on the state of the blockDAG.
A blockchain-consensus game theoretically has an infinite horizon, but here
we consider a finite horizon of some length $T$ unknown to the players.
Following our model
(which we present in Section~\ref{sec:model}), we assume that at 
each node in the game tree player $i$ is in some local state that includes 
their view of the blockDAG and some private information. 
\\
\indent
Following Abraham et al.~\cite{distributed-computing-game-theory}, 
we
consider the following framework.   With each run of the game that results in a
blockDAG $\blockdag$, we associate some \emph{utility} with player $i$, denoted 
$u_i(\blockdag)$. A \emph{strategy} for player $i$ is a (possibly randomized) 
function from $i$'s local state to some set of actions, and tells the player
what to do at each step.
A joint strategy is a \emph{Nash equilibrium} if no player can
gain any advantage by using a different strategy, given that
all the other players do not change their strategies.
An extension of Nash equilibria is \emph{coalition-proof} Nash 
equilibria~\cite{distributed-computing-game-theory}.
A strategy is \emph{$k$-resilient} if a coalition of up to a fraction of $k$ players cannot increase their
utility function by deviating from the strategy, given that other players follow the strategy.
A strategy is \emph{$t$-immune} if, even when a group that comprises a fraction $t$ of the players 
deviate arbitrarily from the protocol, the payoff of the non-deviating players 
is greater than or equal to their payoff in the case where these players do not 
deviate.
A strategy is a \emph{$(k,t)$-robust equilibrium} if it is a $k$-resilient 
and $t$-immune equilibrium.
Similarly, in an $\epsilon$-$(k,t)$-robust equilibrium, players
cannot increase their utility by more than $\epsilon$ by deviating from the
protocol or decrease their utility by more than $1/\epsilon$ in the
presence of a fraction of $t$ irrational players.
\\
\indent
In addition to the players, we assume there exist some other agents that do 
not participate in the game but still maintain a view of the blockDAG.  These
\emph{passive} agents represent full node, and will not accept blocks that are
clearly invalid.  The existence of these agents allows us to 
assume in Section~\ref{sec:sim} that even adversarial players must 
create valid blocks.

%% file: prelim.tex
\subsection{Preliminaries}\label{sec:notation}
If $S$ is a
finite set then $|S|$ denotes its size and $x\randpick S$ denotes sampling a
member uniformly from $S$ and assigning it to $x$.  $\secp\in\N$ denotes the
security parameter and $\usecp$ denotes its unary representation.
\ifndss
\\
\indent
Algorithms are randomized unless explicitly noted otherwise.  PT stands
for polynomial time.  By $y\gets A(x_1,\ldots,x_n;R)$ we denote running
algorithm $A$ on inputs $x_1,\ldots,x_n$ and random coins $R$ and assigning
its output to $y$.  By $y\randpick A(x_1,\ldots,x_n)$ we denote $y\gets
A(x_1,\dots,x_n;R)$ for $R$ sampled uniformly at random. By
$[A(x_1,\ldots,x_n)]$ we denote the set of values that have non-zero 
probability of being output by $A$ on inputs $x_1,\ldots,x_n$.  Adversaries 
are algorithms. \fi We denote non-interactive algorithms using the font
$\mathsf{Alg}$, and denote interactive protocols using the font
$\mathtt{Prot}$.  We further denote such protocols as
$\outputs\randpick\mathtt{Prot}(\usecp,\participants,\inputs)$, where
the $i$-th entry of $\inputs$ (respectively, $\outputs$) is used to denote 
the input to (respectively, output of) the $i$-th 
participant. 
\\
\indent
We say that two probability ensembles $X$ and $Y$ are statistically close over
a domain $D$ if
$\frac{1}{2}\sum_{\alpha\in D}|\pr[X=\alpha]-\pr[Y=\alpha]|$ is negligible; we
denote this as $X\approx Y$.  

%% file: model.tex
\section{Modelling Blockchain Consensus}\label{sec:model}

We present in this section a model for blockchain-based consensus, run amongst
a set of participants (also called players) $\participants$. 
A block $\block$ is considered to be a bet on its ancestors 
$\ancestors(\block)$.  
After introducing the assumptions we make about participants, we
present a model for leader election, which is
used to determine which participants are eligible
to create a block.  We then present a model
for overall blockchain-based consensus, along with its associated security
notions.

\subsection{Assumptions}

We consider a semi-synchronous model~\cite{Dwork:1988} where 
time is divided in units called slots, and each message is delivered within a
maximum delay of $\Delta$ slots (for $\Delta$ unknown to participants).
We assume that all players form a well-connected network, in a way similar to Bitcoin, where they can broadcast a message to their peers, and that communication
comes ``for free.''

\paragraph*{Types of players} We follow the BAR model~\cite{barmodel}, which
means players are either (1) Byzantine, meaning that they behave in an 
arbitrary way; (2) altruistic, meaning that they follow the protocol; or 
(3) rational, meaning that they act to maximize their expected utility.
We use $(\byzfrac,\altfrac,\ratfrac)$ to denote the respective fractions of 
Byzantine, altruistic and rational players.  
Previous research has shown that altruistic behavior is indeed
observed in real-world systems (like Tor or
torrenting)~\cite{distributed-computing-game-theory}, so is reasonable to
consider. 
In addition to these types, as stated in Section~\ref{sec:game-theory} we 
also consider \emph{passive} participants.
They represent the users of a currency who keep a copy of the blockchain and
passively verify 
every block (i.e., the
full nodes). 
They will not explicitly appear in the
protocol, rather we assume that they ``force'' the creation of valid blocks as we explain
in Section~\ref{sec:sim}, since if the chain includes invalid blocks or other obvious forms of misbehavior 
they will simply abandon or fork the currency.
When we say participant or player,
we now mean active player unless specified otherwise.\\
\indent
We consider a semi-permissionless setting, meaning that everyone is allowed to 
join the protocol but they have to place a \emph{security deposit} locking some 
of their funds to do so; if they wish to leave the protocol, then they must 
wait some period of time before they can access these funds again.
This allows us to keep the openness of decentralization while preventing Sybils.
Moreover we consider a flat-model meaning that one participants
account for one \emph{unit} of stake. Thus saying two-third of participants
is equivalent to saying participants that together own two-third of the stake that is in deposit.
Most of the paper makes the assumptions of dynamic committee
where participants can leave and join as explained above. However,
to consider a notion of finality, we will need to strengthen this assumption.
We will detail these assumptions in Section~\ref{sec:protocol},
but briefly we will allow for a reconfiguration period, and assume
that outside of this period the set of participants is fixed.

\subsection{A model for leader election}\label{sec:leader-model}

Most of the consensus protocols in the distributed systems literature
are leader-based, as it is the optimal solution in term of
coordination~\cite{leaderless-consensus}.  Perhaps as a result, leader
election has in general been very well studied within the distributed systems
community~\cite{saks-leaderelection,russel-leaderelection,feige-leaderelection,orv-leaderelection,king-leaderelection}.
Nevertheless, to the best of our knowledge the problem of leader election has
not been given an extensive security-focused treatment, so in
this section we provide a threat model in which we consider a variety of
adversarial behavior.  
%
\\ \indent
Each participant maintains some private state
$\privatestate$, and their view of the
public state $\pubstate$.  For ease of exposition, we assume each
$\privatestate$ includes the public state $\pubstate$.  We refer to a
message sent by a participant as a \emph{transaction}, denoted $\tx$, where
this transaction can either be broadcast to other participants (as in a more
classical consensus protocol) or committed to a public blockchain.
\\ \indent
Our model
consists of three algorithms
and one interactive protocol, which behave as follows:
\begin{description}[itemsep=1pt,leftmargin=0.2cm]
\item[$(\privatestate,\txdep)\randpick\commit(\pubstate)$] is used by a participant to commit themselves to participating
in the leader election.  This involves
establishing both an initial private state $\privatestate$ and a public
announcement $\txdep$.
\item[$\{\privatestate^{(i)}\}_i\randpick\coord(\usecp,\participants,
\{(\round,\privatestate^{(i)})\}_i)$] is run amongst the committed
participants, each of whom is given $\round$ and their own
private state $\privatestate^{(i)}$, in order to update both the
public state $\pubstate$ and their own private states to prepare the leader
election for round $\round$.
\item[$\txrev\randpick\reveal(\round,\privatestate)$] is used by
a participant to broadcast a proof of their eligibility $\txrev$ for round
$\round$ (or $\bot$ if they are not eligible).
\item[$0/1\gets\verify(\pubstate,\txrev)$] is used by a participant
to verify a claim $\txrev$.
\end{description}

We would like a leader election protocol to achieve three security properties:
\emph{liveness}, \emph{unpredictability}, and \emph{fairness}.  The first
property maps to the established property of liveness for
classical consensus protocols, although as we see below we consider several
different flavors of unpredictability that are specific to the
blockchain-based setting.  The final one, fairness (related to 
\emph{chain quality}~\cite{Garay2015}), is especially important in open
protocols like blockchains, in which participation must be explicitly
incentivized rather than assumed.
\\ \indent
We begin by defining liveness, which requires that consensus can be achieved
even if some fraction of participants are malicious or inactive.

\begin{definition}[Liveness]\label{def:liveness}
Let $\afrac$ be the fraction of participants controlled by an adversary $\A$.
Then a leader election protocol satisfies \emph{$\afrac$-liveness} if it is
still possible to elect a leader even in the face of such an $\A$; i.e., if
for every public state $\pubstate$ that has been produced via
$\coord$ with the possible participation of $\A$, it is still possible for at
least one participant, in a round $\round$, to output a value $\txrev$ such that
$\verify(\pubstate,\txrev)=1$.
\end{definition}

Unpredictability requires that participants cannot predict which 
participants will be elected leader before some time.

\begin{definition}[Unpredictability]
A leader election protocol satisfies \emph{unpredictability} if, prior to some
step $\barrier$ in the protocol, no PT adversary $\A$ can produce better than
a random guess at
whether or not a given participant will be eligible for round $\round$, except with negligible
probability.  If
$\barrier$ is the step in which a participant broadcasts $\txrev$, and we
require $\A$ to guess only about the eligibility of honest
participants (rather than participants they control), then we say it
satisfies \emph{delayed} unpredictability.  If it is still difficult for $\A$
to guess even about their own eligibility, we say it satisfies
\emph{private} unpredictability.
\end{definition}

Most consensus protocols satisfy only the regular variant of unpredictability
we define, where $\barrier$ is the point at which the $\coord$ interaction is
``ready'' for round $\round$ (e.g., the participants have completed a
coin-tossing). \ifndss This typically occurs
at the start of the round, but may also occur several rounds beforehand.\fi
\\ \indent
If an adversary is aware of the eligibility of other participants ahead of
time, then it may be able to target these specific participants for a
denial-of-service (DoS) attack, which makes achieving liveness more difficult.
This also helps to obtain security against a fully adaptive adversary that is able to
dynamically update the set of participants it is corrupting. (Since they
do not know in advance which participants to corrupt to gain an advantage.)
A protocol that satisfies delayed
unpredictability solves this issue, however, as participants reveal their
own eligibility only when they choose to do so, by which point it may be too
late for the adversary to do anything. (For example, in a proof-of-stake
protocol, if participants include proofs of eligibility only in the blocks
they propose, then by the time the leader is known the adversary has nothing
to gain by targeting them for a DoS attack. Similarly, an adversary cannot know
which participants to corrupt in advance because it does not know if they will
be eligible.)
\\ \indent
A protocol that satisfies private unpredictability, in contrast, is able to
prevent an adversary from inflating their own role as a leader.  For
example, if an adversary can predict many rounds into the future what their
own eligibility will be, they may attempt to bias the protocol in their favor
by grinding through the problem space in order to produce an
initial commitment $\txdep$ that yields good future results.
\\ \indent
Private unpredictability thus helps to guarantee fairness, which we define as
requiring that each committed participant is selected as leader equally often.
While for the sake of simplicity our definition considers equal weighting of
participants, it can easily be extended to consider participants with respect
to some other distribution (e.g., in a proof-of-stake application, participants
may be selected as leader in proportion to their represented ``stake'' in the
system).


\begin{definition}[Fairness]
A leader election protocol is fair if for all PT adversaries
$\A$
the probability that $\A$ is selected as leader is nearly uniform;
i.e., for all $\round$, $\privatestate$, $\pubstate$ (where again
$\pubstate$ has been produced by $\coord$ with the possible participation of
$\A$), and $\txrev$ created by $\A$,
$
\pr[\verify(\pubstate,\txrev)=1]\approx 1/n.
$
\end{definition}
\saraha{uniform for one participant}

\subsection{Blockchain-based consensus}\label{sec:protocol-definitions}

As with leader election, each participant maintains some private state
$\privatestate$ and some view of the public state $\pubstate$.
We consider the following set of algorithms run by participants:

\begin{description}[itemsep=1pt,leftmargin=0.2cm]

\item[$\privatestate \randpick \setup(\usecp)$] is used to establish the 
initial state: a public view of the blockchain (that is the same for every 
player) and their private state. This includes the $\commit$
phase of the leader election protocol.

\item[$\pi\randpick\eligible(\block,\privatestate)$] is run by each 
participant to determine if they are eligible to place a bet on a block $\block$. 
If so, the algorithm outputs a proof $\pi$ (and if not it outputs $\bot$), that 
other participants can verify using $\verify(\pubstate,\pi)$. The
$\verify$ algorithm is the same as the one used in the leader election
protocol, where $\pi$=$\txrev$.

\item[$\block\randpick\bet(\privatestate)$] is used to create a new 
block and bet on some previous block.

\item[$\block\gets\fcr(\blockdag)$] defines the fork-choice rule that
dictates which block altruistic players should bet on.
To do so, it gives an explicit \emph{score} to different chains,
and chooses the tip of the chain with the biggest score.

\item[$0/1\gets\verifyblock(\blockdag,\block)$] determines whether or not a 
block is valid, according to the current state of the blockDAG.

\item[$\map\gets\lab(\blockdag)$] defines a function $\lab:\trees\to\maps$ 
that takes a view of the blockDAG and associates with every block a label in 
$\{\winner,\loser,\neutral\}$.  This is crucial for incentivization, and is used
to determine the reward that will be associated with every block.   
With each map $\lab$, and player $i$, we associate a utility
function $u_i^\lab$ that takes as input a blockDAG and outputs
the utility of player $i$ for that blockDAG.
We will write $u_i^{\lab}=u_i$ if the label function
is clear from context.
The list of winning blocks constitutes a chain, which we call the \emph{main
chain}.  Every chain of blocks that is not the main chain is called a 
\emph{fork}. \sarahm{Ideally I'd like to find somewhere else to put this, but
it's not a huge priority.}
\end{description}

We would like a blockchain consensus protocol to satisfy a few
security properties.  Unlike with leader election, these have been carefully
considered in the cryptographic
literature~\cite{Garay2015,blockchain-asynchronous}.  
As defined by Pass et al.~\cite{blockchain-asynchronous}, there are four
desirable properties of blockchain consensus protocols: 
consistency, future self-consistency, chain growth, and chain quality.
\\ \indent
Chain growth~\cite{Garay2015} corresponds to the concept of liveness in 
distributed systems, and says the chain maintained by honest players
grows with time. 
Consistency (also called common prefix~\cite{Garay2015}) and future 
self-consistency both capture the notion of safety traditionally used
in the distributed systems literature. Consistency states 
that any two honest players should agree on their view of the chain
except for the last $Y$ blocks, and future self-consistency states that a 
player and their ``future self'' agree on their view of the chain except for 
the last $Y$ blocks (the idea being that a player's chain will not change 
drastically over time).
In our paper, we consider not just a blockchain-based
consensus protocol, but in fact one based on a blockDAG.  Participants thus do
not keep only the longest chain, but all the blocks they receive, which makes
it difficult to use these definitions as is.  Instead, we use the notion of
\emph{convergence},
which states that
after some time,
player converges towards a chain, meaning that no two altruistic players
diverge on their view of the main chain except perhaps for the last blocks,
and that a block that is in the main chain at some time $\tau_0$
will still be in the main chain for any time $t>\tau_0$.
We also ask that this condition holds for a chain of any length, thus capturing
the chain growth (or liveness) in the same definition.
More formally we define convergence as follows:
\begin{definition}[Convergence]
For every $k_0\in\mathbb{N}$, there exists a chain $\chain^0_{k_0}$ of length
$k_0$ and time $\tau_0$ such that: for all altruistic players $i$ and
time $t>\tau_0$: $\chain^0_{k_0}\subseteq \chain^{i,t}$,
except with negligible propability, where $\chain^{i,t}$ denotes the main chain
of $i$ at time $t$.
\end{definition}



Chain quality~\cite{Garay2015} corresponds to the notion of fairness, and says
that honest players contribute some meaningful fraction of all blocks in the 
chain.  

\begin{definition}[Chain quality]
Chain quality, parameterized by $\alpha$, says that an adversary controlling 
a fraction $\coalfrac=\byzfrac+\ratfrac$ of non-altruistic players can 
contribute at most a fraction $\mu=\coalfrac+\alpha$ of the blocks in the main 
chain.
\end{definition}

Finally, we define a relatively unexplored property in blockchain
consensus, \emph{robustness}, that explicitly captures the incentivization 
mechanism. The security notion is not considered by any consensus
protocols (except~\cite{solidus} that states that they leave
the proof for future work) and is paramount to capture the security
of systems where incentives are at the core.\saraha{improve this sentence}

\begin{definition}[Robustness]
A protocol is $\epsilon$-robust if given some fractions 
$(\byzfrac,\altfrac,\ratfrac)$ of BAR players,
following the protocol is a $\epsilon$-$(\ratfrac,\byzfrac)$-robust equilibrium.
\end{definition}

%% file: leaderelection.tex
\section{\sysname: A Leader Election Protocol}\label{sec:construction}
In this section, we present \sysname, a leader election protocol with
minimal coordination that satisfies fairness, liveness, and strong notions of
unpredictability.

\subsection{Our construction}\label{sec:caucus-construction}
\begin{figure*}[t]
\begin{framed}
\centering
\begin{description}[itemsep=1pt]
\item [$\deposit$:] A participant
commits to their \eligibilitytest secret key $\sk$ by creating a $\deposit$ transaction $\txdep$ that
contains the VRF public key $\pk$.  Each broadcast commitment
is added to a list $\commitlist$ maintained in $\pubstate$, and that
participant is considered eligible to be elected leader after some fixed
number of rounds have passed.  

\item [$\coord$:] Once enough participants are committed, participants run
a secure coin-tossing protocol to obtain a random value $R_1$.  They output a
new $\pubstate=(\commitlist, R_1)$.  \textbf{This interactive protocol is run
only for $\round = 1$.}

\item [$\reveal$:] For $\round > 1$, every participant verifies their own
eligibility by checking if $H(y_\round)<\target$,
where $y_\round=G_{\sk}(R_\round)$ and
$\target=\hashmax/\participantsnumber_\round$. (Here
$\participantsnumber_\round$ is the number of eligible participants; i.e., the
number of participants that have committed a sufficient number of rounds before
$\round$ and possibly have not been elected leader in the previous rounds.) The eligible participant, if one
exists, then creates a transaction $\txrev$ with their data
$y_\round$ and $p_\round=p_{\sk}(R_\round)$ and broadcasts it to their peers.

\item [$\verify$:] Upon receiving a transaction $\txrev$ from a participant
$i$, participants extract $y_\round$ and $p_\round$ from $\txrev$ and check whether or not
$\verify_\trm{VRF}(R_\round,y_\round,p_\round)=1$.  If these checks pass,
then the public randomness is updated as $\globalrandao_{\round+1}\gets
\globalrandao_\round\oplus y_\round$ and they output $1$, and otherwise
the public state stays the same and they output $0$.
\end{description}
\end{framed}
\caption{The \sysname protocol.}
\label{fig:construction}
\end{figure*}
The full \sysname protocol is summarized in Figure~\ref{fig:construction}.  
Our construction is similar to that of Algorand~\cite{algorand}.
We, however, add a secure initialization of the random beacon and use Verifiable
Random Delays to achieve liveness.\\
\indent
\ifndss
We assume participants have generated signing
keypairs and are aware of the public key associated with each
other participant (which can easily be achieved at the time participants
run $\deposit$).  We omit the process of generating keys from our formal
descriptions.
\\ \indent
\fi
To to be considered as potential leaders, participants must place a
\emph{security deposit}, which involves creating a commitment to their
\eligibilitytest secret key.  
This means the $\deposit$ function runs $\gen$ and returns the 
\eligibilitytest secret key as the private state of the participant and the \eligibilitytest public key 
(incorporated into a transaction) as the transaction to add to a list of
commitments $\commitlist$ in the public state.
\\ \indent
In the first round, participants must interact to establish a shared random
value.  This can be done by running a coin-tossing protocol
to generate a random value $\globalrandao_1$.  We suggest using 
SCRAPE~\cite{scrape}, due to its low 
computational complexity and compatibility with public ledgers.  Any solution
that instantiates a publicly verifiable secret sharing (PVSS) scheme, however,
would also be suitable. The only requirement that we have is that the PVSS
should output a value that is of the same type as the value output by the VRF function $\pef$.
\\ \indent
For each subsequent round, participants then verify whether or not they
are eligible to fold their randomness into the global value by
checking if $H(G_{\sk}(\globalrandao_{\round}))
< \target$, where the value of $\target$ depends on the number of
expected leaders per round (for example, we choose $\target=\hashmax/\participantsnumber_\round$
in order to have on expectation one leader per round).
If this
holds, then they reveal $y_\round=G_{\sk}(\globalrandao_{\round})$
and $p_\round=p_{\sk}(\globalrandao_{\round})$ to the other participants, who can
verify that the inequality holds and that $\verify_{\trm{VRF}}(\globalrandao_{\round},y_\round,p_\round)=1$.
If the participant is in fact eligible,
then they are deemed to be the leader for that round and the global 
randomness is updated as $\globalrandao_{\round+1}\gets
\globalrandao_{\round} \oplus y_\round$.
\\ \indent
To fully achieve security, we describe two necessary alterations to the basic
protocol as presented thus far.  First, in order to maintain unpredictability,
participants should become eligible only after some fixed number of
rounds have passed since they ran $\commit$.  This means updating $\verify$ to
also check that
$\round_\mathsf{joined} > \round - x$ (where $\round_\mathsf{joined}$ is the
round in which the participant broadcast $\txdep$ and $x$ is the required
number of interim rounds).
This has the effect that an adversary 
controlling $\afrac$ participants cannot privately predict that they will be
elected leader for a few rounds and then grind through potential new secret key
values to continue their advantage via new commitments, as the probability 
will be sufficiently high that at least one honest participant will be elected 
leader between the time they commit and the time they participate.
\\ \indent
Second, it could be the case that in some round, no participant is elected
leader. It could also be the case that an adversary is the only elected leader
and does not reveal their proof of eligibility and abort.
To maintain liveness, we alter the protocol so that if no participant 
reveals $\txrev$, we ``update'' the random beacon 
as $\globalrandao_\round\gets F(\globalrandao_\round)$, where $F$ is a
deterministic function that acts as a 
\emph{proof-of-delay}~\cite{proof-of-delay,vdf}.
The simplest example of such a
function is $F = H^p$; i.e., a hash function iterated $p$ times.
When a honest player is not eligible on the winning block,
they start computing the proof-of-delay and if by the time it is computed
no leader has been revealed, they check their eligibility with 
the updated beacon $F(\globalrandao_\round)$.
One can think of this as 
re-drawing the lottery after some delay. 
This allows
participants to continue the protocol (and, since it is purely deterministic,
is different from proof-of-work), but has the downside that verification is
also costly, as it requires participants to re-compute the hashes.
 B{\"u}nz et al. recently proposed an efficient Verifiable Random Delay
 function~\cite{vdf} that can be efficiently and publicly verified yet requires
sequential steps to compute.\saraha{add other citations since there are three different now}
This
proof-of-delay should be used sparingly so we consider it
acceptable that it is expensive.
Moreover we make the assumptions that the time $\delta$ that it takes
to compute the proof-of-delay is bigger than the parameter $\Delta$
of our semi-synchronous model.
\\ \indent
It could also be the
case, however, that there are two or more winners in a round. \ifndss
In a setting such as proof-of-stake, being
elected leader comes with a financial reward, and conflicts may arise if two
winners are elected (such as the nothing-at-stake problem discussed in
Section~\ref{sec:defns-pos}). One potential solution (also suggested by
Algorand~\cite{algorand}) for electing a single leader is to
require all participants to submit their
winning values $y_{\round}$ and then select
as the winner the participant whose pre-image $y_{\round}$ has the lowest bit
value. \fi This problem is investigated further in Section~\ref{sec:protocol}, 
where we present the full \betting protocol.
\\ \indent
Finally, we describe a third, optional alteration designed to improve
fairness by forcing more rotation amongst the leaders.  To be
elected, we require that a participant has not acted as leader in the past
$(\participantsnumber_\round-1)/2$ rounds, which can be implemented by adding a condition in the
$\verify$ function and using $(\participantsnumber_\round+1)/2$ in place of 
$\participantsnumber_\round$ in the computation of the value $\target$.

\subsection{Security}\label{sec:leader-security}
In order to prove the security of \sysname as a whole, we first prove the
security of its implicit random beacon.
\begin{lemma}\label{thm:randao}
If $H$ is a random oracle and $\globalrandao$ is initialized using a secure
coin-tossing protocol, then the random beacon $R_\round$ is also secure; i.e., 
it satisfies liveness (Definition~\ref{def:coin-liveness}),
unbiasability (Definition~\ref{def:coin-unbiasability}), and unpredictability
(Definition~\ref{def:coin-unpredictability}) for every subsequent round.
\end{lemma}
\begin{proof}
For liveness, we observe that after initialization, no coordination
is required, so any online participant can communicate their own eligibility
to other online participants, allowing them to compute the new random value.
The exception is the case where no participant is elected leader, in which case participants can update their random value by
$R_\round\gets F(R_\round)$ until a leader reveals themselves.
\\ \indent
For unpredictability, we must show that, unless the adversary is
itself the next leader, it is hard to learn the value of
$\globalrandao_\round$ before it receives $\txrev$.  We have
$\globalrandao_\round=\globalrandao_{\round-1}\oplus y_{\round-1}$,
where $\globalrandao_{\round-1}$ is assumed to be known.  In the protocol, 
the adversary sees $\pk$ as part
of the commitment of the relevant honest participant, and if that participant
has run $\reveal$ before it may have also seen $y_\round'=G_{\sk}(R_\round')$
for $\round' < \round$. %
If the adversary could produce better than a random guess about $\globalrandao_\round$
then they would also produce better than a random guess about $G_{\sk}(R_\round)$
which contradicts its pseudo-randomness.
\saraha{do we want to write a reduction here?}
\\ \indent
For unbiasability, we proceed inductively.  By assumption, $\globalrandao$ is
initialized in a fair way, which establishes the base case.  Now, we assume
that $\globalrandao_{\round-1}$ is uniformly distributed, and
would like to show that $\globalrandao_\round$ will be as well.  By the 
assumption that $R_\round$ is unpredictable, and
thus unknown at the time an adversary commits to their secret key, the distribution
of $y_\round$ and $R_\round$ is thus independent. (As the adversary cannot
grind through private keys since they do not know $R_\round$ at the
time they commit and G verifies unicity.)
%
If we define a value $R$, and
denote $R'\gets R\oplus y_{\round-1}$, then we have
\ifndss{
\begin{align*}
\pr[R_\round = R] &= \pr[y_{\round-1}\oplus R_{\round-1} = y_{\round-1}\oplus R']\\
                   &= \pr[R_{\round-1} = R'],
\end{align*}}
\fi
which we know to be uniformly random by assumption, thus for every $R'$, we have
$\pr[R_{\round-1} = R']\approx 1/ (2^\ell-1)$ and for every value
$R$, $\pr[R_\round = R]\approx 1/ (2^\ell-1)$, proving the fairness of 
$R_\round$. (Here $\ell$ denotes the bitlength of $R$.)
An adversary controlling many winning values for $R_{\round-1}$, however, may 
decide on the one they reveal, which makes it difficult to argue for
unbiasability when the adversary controls many players.
This is a limitation of \sysname that we address (or at least quantify) in 
our analysis of \betting in Section~\ref{sec:sim}.

\end{proof}

\begin{theorem}\label{thm:caucus}
If $H$ is a random oracle and $R$ is initialized as a uniformly random value,
then \sysname is a secure leader election protocol; i.e., it satisfies
liveness, fairness, delayed unpredictability (where
$\barrier$ is the step at which the elected leader reveals their proof), and
private unpredictability (where $\barrier$ is the step at which the
randomness $\globalrandao_\round$ is fixed).
\end{theorem}

\begin{proof}
For liveness, a participant is elected if they broadcast a valid transaction
$\txrev$ such that $H(y_\round)<\target$. If $\coord$
satisfies $\afrac$-liveness then an adversary controlling $\afrac$ participants
cannot prevent honest participants from agreeing on $R_\round$.
In the case where no participants produce a value $y_\round$ such that
$H(y_\round)<\target$, we update the value of $R_\round$ as
described above until one participant is elected.  With a similar argument as
in the proof of Theorem~\ref{thm:randao}, the protocol thus achieves liveness.
\\ \indent
For fairness, a participant wins if $H(G_{sk}(R_\round)) <
\target_\round$.  By the assumption that $R_\round$ is unpredictable, and
unknown at the time an adversary commits to their secret, we have that
an adversary could not have grind through secret keys to bias $G_{sk}(R_\round)$.
Combining this with the assumption that $R_\round$ 
is unbiasable and thus uniformly distributed, we can argue that 
$y_\round=G_{sk}(R_\round)$ is uniformly distributed. This implies that the
probability that the winning condition holds is also uniformly random, as
desired.
\\ \indent
The argument for delayed unpredictability is almost identical to the
one in the proof of
Theorem~\ref{thm:randao}: even when $R_\round$ is known, if $\A$
has not formed $y_\round$ itself then by the pseudo-randomness of G
it cannot predict whether $H(G_{\sk}(R_\round)) < \target_\round$.
(If it could then the adversary could use that to distinguish $G_{\sk}(R_\round)$
from random with an advantage.)
\\ \indent
Finally, private unpredictability follows from the unpredictability of
$R_\round$.
\ifndss
\\ \indent
In terms of the fraction $\coalfrac$ of malicious participants that we
can tolerate, it is $t / n$, where $t$ is the threshold of the PVSS scheme
used to initialize the random beacon.
In the context of proof-of-stake, 
however, we would still need to assume an honest majority. We investigate this
in the next section, where we present the full \betting protocol. \fi
\end{proof}
Even if the initial value was some constant instead of a randomly
generated one, we could still argue that the protocol is fair after the
point that the randomness of at least one honest participant is incorporated
into the beacon. 
This assumption is weakened the longer the
beacon is live, so works especially well in settings 
where the leader 
election protocol is used to bootstrap from one form of consensus (e.g., PoW) 
to another (e.g., PoS), as discussed for Ethereum~\cite{casper}.

%% file: protocol.tex
\section{\betting: A Consensus Protocol}\label{sec:protocol}
In this section, we present \betting, our full blockchain consensus protocol.
Our focus is on incentives, and in particular on enforcing good behavior even
in settings where no natural incentives or investments exist already.
\\ \indent
Briefly, \betting works as follows: participants bet on the block that has the 
strongest \emph{score}, according to their view of the blockDAG.  They also 
reference all the leaves (i.e., the most recent blocks) of which they are 
aware, to ``prove'' that they are well connected and following the rules.  A 
block is valid if among all of its references, it is indeed betting on the 
one with the higher score.
We argue for its security more extensively in the next section, but give here
some intuition for how it addresses the challenges presented in the PoS
setting introduced in Section~\ref{sec:defns-pos}.  First, \betting solves the 
nothing-at-stake problem by strongly punishing players who do not reference
their own blocks, and ensuring that players bet only on the strongest chain
they see.  As it is not possible for two chains to appear as the strongest at
the same time, this prevents players from betting on multiple chains.
\\ \indent
The grinding attack is prevented largely due to the unbiasability of the 
random beacon in \sysname, and the optional requirement that leaders must 
rotate with sufficient frequency discussed at the end of
Section~\ref{sec:caucus-construction}.
Finally, long-range attacks are thwarted by \betting's finality rule, which
acts as a form of decentralized checkpointing.
\subsection{Protocol specification}\label{sec:protocol-rule}

We specify how to instantiate the algorithms required for a
consensus protocol specified in Section~\ref{sec:protocol-definitions}.  The
$\setup$ and $\eligible$ algorithms are as described for \sysname in
Section~\ref{sec:construction}.  Before presenting the rest of the algorithms, we
give some additional definitions associated with finality in blockDAGs.
\\ \indent
\saraha{maybe move this somewhere else}
In order to make the following definitions, we assume a static set of
players. We then present how to handle a dynamic set.
Let's also recall that thanks to the deposit, the set of players is known to everyone.
A \emph{candidate} block is a block, as its name suggests, that is a 
candidate for finality.  The initial list of candidate blocks consists of 
every block \emph{betting} on the genesis block; i.e., every block that uses 
this as its parent $\blockprev$.
Whenever a block has in its past two-thirds of bets on a candidate,
this block acts as a \emph{witness} to the
finality of that block, so is called a witness block.  If a block $\bone$ is a 
witness for $\block_0$ and $\btwo$ is a witness for $\bone$, we say that 
$\btwo$ is a \emph{second witness} of $\block_0$.  Finally, candidate blocks
belong to a given \emph{rank}, which we denote by $\rank$.  The first set of 
candidate blocks (after the genesis block) belong to $\rank=1$. 
After this 
every block that bets on a second witness block of rank $\rank$
and has a distance of $\mathsf{w}$ with this block is a candidate for rank 
$\rank+1$.
The above process constitutes a decentralized checkpointing.
\\ \indent
Using the sample blockDAG in Figure~\ref{fig:blockdag}, for example, and 
assuming three players and $\mathsf{w}=0$, the initial list of
candidate blocks is $\{A, B, C, D\}$ (which all have rank $\rank=1$).  Block $E$ 
then bets on $A$, as does
block $H$ (since $A$ is an ancestor of $H$), so $H$ can be considered
as a witness block for $A$.  Similarly, $I$ acts as a witness block for
$C$.  Since this now constitutes two-thirds of the participants, $A$ and $C$
become \emph{justified}~\cite{casper}.  If in turn two-thirds of participants
place bets on the associated witness blocks, then the justified block becomes
\emph{finalized}. 
With $\mathsf{w}=0$, the second witness blocks are added to the candidate list accordingly, 
but at rank $\rank=2$.
\\ \indent
In order to handle a dynamic set of players,
once a block is finalized (i.e. once there exists
at least one second witness block)
we allow for a window
of $\mathsf{w}$ blocks for the players to leave or join the protocol.
At the end of this period, the set of players is fixed and the decentralized
checkpointing resumes with as new candidate blocks, all blocks that have a distance
$\mathsf{w}$ with a second witness block for rank $\rank$.
We leave as important future work a solution that would allow players
to leave and join the protocol even during the checkpointing period.

\paragraph*{Fork choice rule}
We present a formal specification of our fork choice rule $\fcr$ in 
Algorithm~\ref{alg:fcr}.  Intuitively, the algorithm chooses blocks with more 
connections to other blocks.  Accordingly, we compute the score of a leaf 
block $\block$ by counting the number of 
outgoing references for every block in its past. We do not count, however,
blocks that have been created by an adversary using the same eligibility proof
multiple times.
We denote as $\double$ the set of all blocks that contains
the same proof of eligibility but different content.
The score of a chain whose tip is $\block$ is then the number of edges 
in the subgraph induced by $\past(\block)\setminus \double$,
and we pick as a ``winner'' the
leaf with the highest score.  If there is a tie (i.e., two leaf blocks have 
the same score), we break it by using the block with the
smallest hash.\footnote{It is important, to avoid grinding attacks, to use the 
    hash as defined in \sysname, or something else similarly unbiasable.}

\begin{algorithm}
\SetKwInOut{Input}{input}
\SetKwInOut{Output}{output}
\Input{a DAG $\tree$} 
\Output{a block $\block$ representing the latest ``winner''}
\If{$(\tree=\textsf{Genesis Block})$} {return $\tree$}
$w\gets\emptyset$ \\
\For{$\block\in\dagleaves(\tree)$}{
\For{$\block'$ in $\past(\block)\setminus \double$}{
$w[\block'] = |\block'[\leaves]|$ 
}
}
$\maxweight\gets \trm{argmax}_{\block\in \leaf(\tree)} w(\block) $ \\
\tcp{{\small if there is a tie choose block with smaller hash}}
$\block\gets \trm{argmin}_{\block\in\maxweight} H(\block)$\\
return $\block$
\caption{Fork-choice rule (FCR)\label{FCR}}
\label{alg:fcr}
\end{algorithm}

\paragraph*{Betting}
To place a bet, a participant first identifies the latest winning block as  
$\block\gets\fcr(\treepub)$. 
They then check their latest second witness block (if they have one)
and verify that at least one candidate block associated with it is also
in $\ancestors(\block)$ (i.e. they verify that the block was not created
maliciously as part of a long range attack).
They then check to see if they are eligible to
act as a leader by computing $\pi\randpick \eligible(\block,\privatestate)$.
If they are (i.e., if $\pi\neq\bot$), then they form a block with $\block$ as
the parent, with all other blocks of which they are aware as the leaf blocks, 
and with their proof of eligibility $\pi$ and set of transactions.

\paragraph*{Block validity}
We now define the rules that make a block valid; i.e., the checks
performed by $\verifyblock(\blockdag,\block)$.  Intuitively, a valid block 
must be betting on the block chosen by the fork-choice rule, and its creator 
must be eligible to bet on that block.  If a player
is aware of a justified block, then they must bet on either that
block or another witness block, but cannot prefer a non-justified block to 
a justified one.

More formally, a new block $\block =
(\blockprev,\leaves,\pi,\txset)$ is valid only if the following hold:

\begin{enumerate}[leftmargin=0.2cm]
\item It is betting on the block chosen by the fork choice rule for 
the blocks of which it is aware; i.e., $\blockprev = \fcr(\past(\block))$.

\item The creator is eligible to bet: 
$\eligible(\blockprev, \privatestate^{\block.\sender}) \neq \bot$.

\item If it references a witness block then it is betting on a witness 
block; i.e., if there exists a witness block in $\past(\block)$
then there exists a witness block in $\ancestors(\block)$. 

\item If it references a second witness block, then it is betting on a block 
in the past of that block: if there exists a second
witness block $\block_s\in\past(\block)$, then 
$\ancestors(\block)\cap\past(\block_{s})\ne\emptyset$.

\end{enumerate}


\subsection{Incentives}\label{sec:incentives}

\paragraph*{Label}
We present a formal specification of our label function $\lab$ in 
Algorithm~\ref{alg:label}.  
Intuitively, if a block is chosen by the $\fcr$ it is labelled $\winner$
and so are all of its ancestors. Blocks that bet on winners are 
labeled $\neutral$.  Following the techniques in PHANTOM~\cite{phantom}, all 
winning and neutral blocks form a subset of the DAG called the \emph{blue} 
subset and denoted by $\blue$.  Every block whose anticone intersects with 
fewer than $k$ blocks in the blue set is labeled $\neutral$, and otherwise it 
is labeled $\loser$.  The parameter $k$ is called the \emph{inter-connectivity}
parameter, and means that a block is allowed to be ``unaware''
of $k$ winning blocks, but not more (as, e.g., these blocks may have been 
created at roughly the same time).

\begin{algorithm}
\SetKwInOut{Input}{input}
\SetKwInOut{Output}{output}
\Input{A DAG $\tree$} 
\Output{A labelling of the block in the DAG $\map$}
set $\block\gets\fcr(\tree)$ \\
$\blue\gets \blue \cup \{\block\}$\\
$\map(\block)=\winner$\\
\For{$\block_i \in \ancestors(\block)$}{
$\blue\gets \blue \cup \{\block_i\}$\\
$\map(\block_i)=\winner$\\
\For{$\block_j\in\directfuture(\block_i)\setminus\ancestors(\block)$}{
$\blue\gets \blue \cup \{\block_j\}$\\
$\map(\block_j)=\neutral$
}
}
\For{$\block_i\in\tree\setminus\blue$}{
\eIf{$\anticone(\block_i)\cap\blue\le k$}{
$\blue\gets \blue \cup \{\block_i\}$\\
$\map(\block_j)=\neutral$
}
{$\map(\block_i)=\loser$}
}
\Return $\map$
\caption{$\lab$}
\label{alg:label}
\end{algorithm}

\paragraph*{Utility functions}
At the end of the game, which we define to be of length $T$, as defined
in Section~\ref{sec:model}, we take the $\bcpc$ and apply the $\lab$ function to it,
in order to associate each block with a state.
For every winning block that a player has added to the $\bcpc$, they win a 
reward of $\rwd(\block)$, and for every losing block they lose $\pun$.
In addition, if a player creates a block that does not reference one of
their own blocks, we add a bigger punishment $\bigpun$ (for example,
blocks that belong to the set $\mathsf{Double}$ defined previously
will add this punishment).  This punishment is
bigger because a player is obviously aware of all their own blocks, so if they
do not reference one it is an obvious form of misbehavior (whereas a block
might end up being labelled a loser for other reasons).
\\ \indent
More formally, we define the following utility function:
\begin{multline}u_i(\bcpc)=\sum_{\substack{\block \in \bcpc \trm{ s.t.}\\
\block.\sender=i\trm{ and } \\M(\block)=\winner}}\rwd(\block)
-\sum_{\substack{\block \in \bcpc \trm{ s.t.}\\ \block.\sender=i\trm{ and }
\\M(\block)=\loser}}\pun \\
+\bigpun\times |N|
\label{eq:utility}
\end{multline}
where $M = \lab(\bcpc)$ and
\ifndss{
\begin{align*}
N=\{(\block_j,\block_k) \trm{ s.t. } &\block_j.\sender=i \ \land\ 
    \block_k.\sender=i\ \land \\
&\block_j\notin\past(\block_k)\ \land\
\block_k\notin\past(\block_j)\}
\end{align*}}
\fi

The reward function is proportional to the connectivity of a block; i.e., a block 
that references many blocks receives 
more than a block that references only one other block.
The reason is that we want to incentivize players to exchange
blocks between each other, rather than produce blocks privately (as in a
selfish mining attack).  
In this paper we 
consider a simple function $\rwd(\block)=|\leaves|\times c$ for some constant
$c$, and treat $\pun$ and $\bigpun$ as constants.  We leave the study of more 
complex reward and punishment mechanisms as interesting 
future work.  
\\ \indent
One of the difficulties of dealing with blockchain-based consensus, compared to 
traditional protocols, is that the enforcement of the payoff is achieved only 
by consensus; i.e., the utilities depend on whether or not enough players 
enforce them.  In order to enforce the payoff, we thus assume that participants 
can give a reward to themselves in forming their blocks (similarly to Bitcoin), 
but that evidence of fraud can be submitted by other players.
If another player submits evidence of fraud, the subsequent punishment
is taken from the security deposit of the cheating player.

%% file: simulation.tex
\section{Security of \betting}\label{sec:sim}
In this section, we show that \betting is secure, according to the model in 
Section~\ref{sec:model}.  We support our proofs with a simulation 
of the protocol as a game played between Byzantine, altruistic, and rational 
players in Section~\ref{sec:simulations}.  

\subsection{Action Space}\label{sec:sim-settings}
We discuss the different strategies available to
a coalition of players, whether Byzantine or rational. They
can take any deviation possible from the game.  We do assume,
however, that they create valid blocks, as otherwise passive players will
simply ignore their chains (as discussed in Section~\ref{sec:model}).
\\ \indent
If an adversary withholds their blocks, it can gain an advantage in
subsequent leader elections.  To see this, consider that after each block 
each player has a probability $1/n$ of being elected leader.  Being a leader 
does not guarantee a winning block, however, as only the block with the 
smallest hash wins.  For each block, the number of subsequent players $k$ that 
are elected leader follows a binomial distribution parameterized by $n$ and 
$1/n$.  Assuming the leader election is secure, each of these leaders is 
equally likely to have the winning block, so each player has probability 
$1/(kn)$ of being the winner. By not revealing a block, this probability goes 
up to $1/n$ (since other players are simply not aware of it), so by keeping 
their chain private an adversary can raise their 
chance of having a winning block.  We thus assume that both Byzantine and
rational players withhold their blocks and grind through all possible
subsequent blocks in order to maximize their advantage.
\\ \indent
In terms of the space that players grind through, the main option they have
when elected leader 
is whether to place a bet or not.  The protocol dictates that they must bet on the
fork-choice rule only, but they may wish to bet on a different block (as,
e.g., doing so could increase their chances of being elected leader in the
future).  In order to still maintain the validity of their blocks, doing so
means they must eliminate references in their set $\leaves$ so that their
chosen block appears as the fork-choice rule (in accordance with the first
check in $\verifyblock$).
Players are always better off, however, including
as many blocks as possible in their references, as it increases the score of
their block.  Thus, they will remove references to blocks that have higher 
scores in order to make their block appear as the fork-choice rule, but not 
more than necessary.
\\ \indent
For rational players, there is a trade-off between not revealing their block
(which raises their chance of having more winning blocks, as argued above)
and revealing their block, which reduces their chance of having more winning 
blocks but increases their reward because it allows their block to have more 
references.  Our simulation investigates this trade-off.
Regardless, the strategy of
rational players is to grind through all possible blocks and broadcast the
chain that maximizes their utility.
Byzantine players, in contrast, do not try to maximize their profit, but
instead play irrationally (i.e., they are not deterred by punishment). 

%% file: security.tex
\subsection{Security arguments}\label{sec:sim-security}
According to the security properties in
Section~\ref{sec:protocol-definitions}, we need to argue three things:
convergence, chain quality, and robustness.  We
support these security properties with our
simulation in Section~\ref{sec:simulations}
\iffull\else{Due to space constraints, full proofs for all of our theorems 
can be found in Appendix~\ref{sec:proofs}.  Nevertheless, we state and provide
some intuition for the theorems here.}\fi 
\newcommand\finality{
We show in Theorem~\ref{thm:finality} that for a given rank, all second witnesses
share the same set of candidate blocks
provided that $f_B<1/3$.
The idea is that whenever a block is finalized, all players have agreed on the current
set of candidate blocks, and thus they should not accept any other candidate blocks at
that rank (as required by the betting algorithm stated in Section~\ref{sec:protocol}). 

\begin{figure}[t]
\centering
\includegraphics[width=0.6\linewidth]{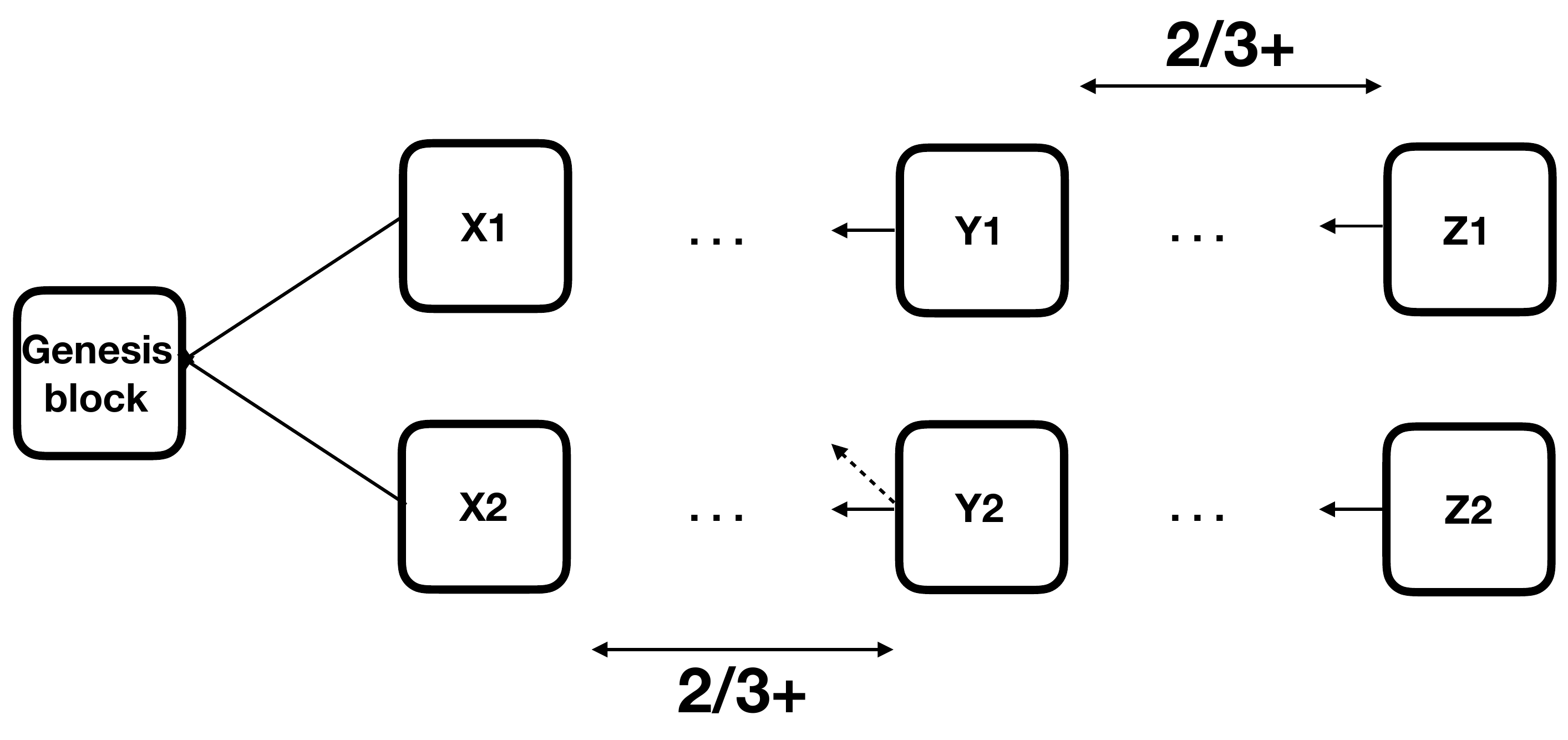}
\caption{A visual sketch of the proof of Theorem~\ref{thm:finality}.
A participant placing a bet between $x_2$ and $y_2$, and 
$y_1$ and $z_1$ must have placed their bet on $x_2$ first, thus $z_1$ 
references $x_2$.}
\label{fig:blockdag-fin}
\end{figure}

\begin{theorem}[Finality]\label{thm:finality}
If $\byzfrac<1/3$,
then once a block at rank $\rank$ is finalized,
players 
agree on a list of candidate blocks for rank $\rank$ (i.e., this list cannot 
grow anymore).
\end{theorem}
\begin{proof}
Assume there exist two finalized blocks $x_1$ and $x_2$.  Denote by $y_i$ the 
witness block for $x_i$, and by $z_i$ the second witness block for $x_i$.
By the definition of finality, it must have been the case that more than 
two-thirds of participants placed bets on $x_1$ and $x_2$, which in turn
implies that more than a third of them placed bets on \emph{both} $x_1$ and $x_2$. 
Since only a third of participants are Byzantine,
this means that at least one non-Byzantine player placed a bet on both blocks.
(Let's recall that during the decentralized checkpointing the set of
players is fixed.)
Non-Byzantine players always reference their own block (since there is a large 
punishment incurred if not); this means that there is one of
the witness blocks $y_i$ that references ``across the chains''; i.e., such that
$x_1,x_2\in\past(y_i)$. 
\\ \indent
Without loss of generality, assume that $x_1,x_2\in\past(y_2)$. This means 
that $x_1\in\past(z_2)$ and thus $x_1,x_2\in\past(z_2)$.
Since more than two-thirds of the participants bet on $y_1$, by a similar 
reasoning as above this means that at least one non-Byzantine player placed 
a bet on both $y_1$ and $x_2$ (before $y_2$).  Thus, either (1) $z_1$ 
references $x_2$ or (2) $y_2$ references $y_1$.  Because of the third check
in $\verifyblock$, however, 
this second case is not possible: since $y_1$ is a witness block, $y_2$ cannot
reference it without betting on a justified block, and $x_2$ is not justified 
before $y_2$.  
We refer the reader to graph~\ref{fig:blockdag-fin} for a visual intuition of 
this.  Thus, it must be
the case that $z_1$ references $x_2$.
Thus $z_i$ references both $x_i$ and $x_j$ for $i\ne j$.
\\ \indent
For a set of candidate blocks, all the second witness blocks thus reference all 
candidate blocks (applying the previous analysis to all the candidates blocks 
pair-wise).  So, after a candidate block is finalized, players must agree on the set 
of candidate blocks for that rank.
\\ \indent
We will investigate the long-range attack in the proof of the convergence
theorem.
\end{proof}
}
\newcommand\generalclaim{
Before proving each security property, we first prove some general results
about the protocol.
Assume that two blocks $\bone$ and $\btwo$ are two competing blocks with the 
same score, betting on the same block. 
We call the stronger chain the one with the higher score according to a 
hypothetical oracle node that collates the views of the blockDAG maintained
by all participants.  We assume that $\bone$ is the leaf of the stronger chain.
\begin{claim}
On average, blocks added to the stronger chain add more to its 
score than blocks added on the weaker chain.
\label{claim:blockweight}
\end{claim}
\begin{proof}
\ifndss{
Without loss of generality, we prove the result for two chains with the same 
score $S$.  For simplicity, we assume that the chains do not reference
each other (e.g., chains $\{A,C\}$ and $\{B,D\}$ in Figure~\ref{fig:blockdag-weight}).}
\fi
By the definition of scoring in the $\fcr$ (Algorithm~\ref{alg:fcr}), a block betting
on the stronger chain adds a score of $S+2$ if the leader is aware of
the leaf block on the weaker chain, 
and $S+2-1=S+1$ otherwise.  (For example, $E$ adds a score of 
$S+2$ by referencing both $C$ and $D$, and $S+1$ by referencing $C$ and $B$ 
if it is not aware of $D$.)
The block that they receive first has probability
0.5 of being stronger or weaker as the random beacon has uniform distribution, thus
on average, the score added by one block to the stronger chain is 
$0.5\cdot(S+2)+0.5\cdot(S+1)=S + 1.5$.
\begin{figure}[h]
\centering
	\includegraphics[width=0.6\linewidth]{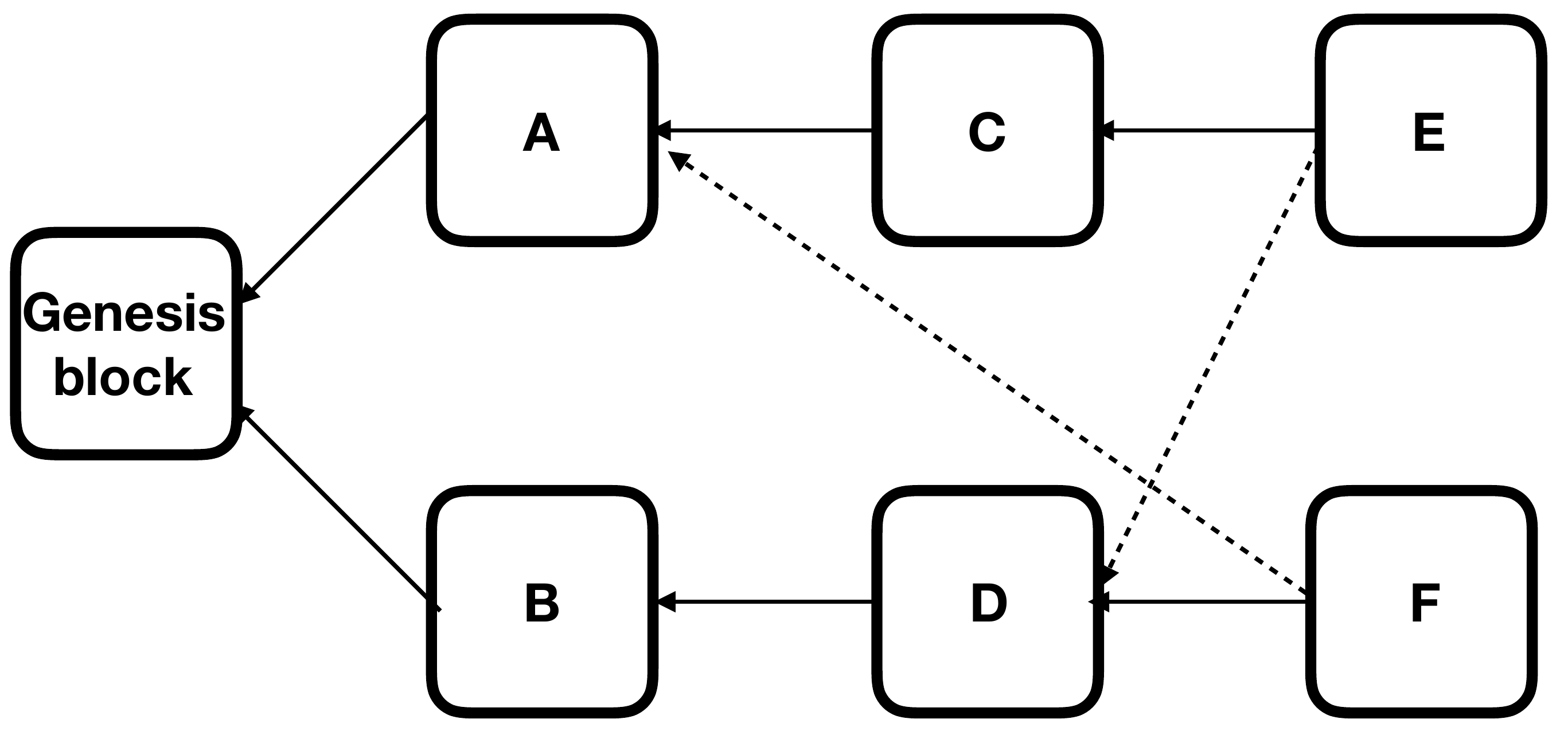}
	\caption{Example of two blocks added on competing chains.}
	\label{fig:blockdag-weight}
\end{figure}
On the weaker chain, in contrast, whether the leader is aware of the stronger leaf block 
or not does not matter since they cannot reference it while maintaining the validity of
their block (according to the first check in $\verifyblock$).
Thus the score added to the weaker chain is 
$(S-1)+2=S+1$, since that block references the stronger chain 
but not on its latest block.
\end{proof}
\finality
}
\iffull\generalclaim\fi
Next, we show convergence. 
Intuitively, the main argument here is that if 
an adversary tries to grow multiple chains to 
prevent altruistic players from agreeing on a main chain, 
altruistic players are still very likely to agree on a chain, since the score 
of the main chain grows faster than the score of other chains. 
Additonally, since an adversary is unlikely to be elected leader for a 
consecutive number of blocks, it is unlikely that other players 
will revert their main chain once they have agreed on one.  
Even if the adversary somehow manages to build their own private chain (by,
e.g., using the proof-of-delay or bribing old participants in a long-range
attack), the decentralized checkpointing mechanism in \betting provides a
notion of finality for blocks.  Thus, by the time they succeed in mounting 
such an attack, it will be too late and other participants will not accept 
their chain.
\begin{theorem}[Convergence]\label{thm:fork}
Given a coalition $\coalfrac<1/3$ of non-altruistic players, 
we have:
for every $k_0\in\mathbb{N}$, there exists a chain $\chain^0_{k_0}$ of length
$k_0$ and time $\tau_0$ such that: for all altruistic players $i$ and
times $t>\tau_0$: $\chain^0_{k_0}\subseteq \chain^{i,t}$,
except with negligible propability, where $\chain^{i,t}$ is the main chain
of player $i$ at time $t$.
\end{theorem}

\newcommand\forkresistance{
\begin{proof}
We start by showing that the length of the longest chain in the DAG, indeed grows.
This is relatively straightforward, and follows from the discussion around
liveness in Section~\ref{sec:protocol-rule}.  To summarize, even if a
Byzantine player is the only leader and chooses not to publish a block, after
some delay players will ``re-draw'' the lottery and an altruistic player will be 
chosen eventually.
We now calculate the worst-case growth rate of the main chain, which 
happens when the adversary simply aborts.  If an adversary controls $\ncoal$ players, then an
aborting player is elected leader with probability $1-(1-p)^{\ncoal}$. This means that after
each block, the chain will grow normally with probability $(1-p)^{\ncoal}$ and will
grow with a delay of $\delta$ (where $\delta$ is the delay in the 
proof-of-delay) with probability $1-(1-p)^{\ncoal}$.
Moreover since a block is propagated with a maximum delay of $\Delta$, we have
that the worst rate at which a block is created is $(1-p)^{\ncoal}\cdot \Delta +(1-(1-p)^{\ncoal})\cdot (\Delta+\delta)=\Delta + \delta\cdot (1-(1-p)^{\ncoal}) $.
However it could be the case that more than one chain grows in the DAG.
We now move on to prove the core of the protocol.
\\ \indent
Let $k_0$ be an integer.
Due to the previous argument about the growth of the chain
and the semi-synchrony assumption (all players receive a block
before $\Delta$ slots), there exists a time
$\tau_1$ such that every honest player have
in their DAG at least one chain $C_{k_0}$ of length $k_0$.
Let's assume that there exists two such chains $C_{k_0}$ and $C'_{k_0}$.
We show that, with high probability, after some time $\tau_0$ one will be ``dropped'' and
thus for the remaining one we will have that
for every $t>\tau_0$ $C^0_{k_0}\subseteq \chain^{i,t}$.
\\ \indent
Let's assume that players start creating blocks on both chains.
After some time less or equal than $\Delta$ players will be aware of the
other chain and thus can start referencing it (remember that $\Delta$
is smaller than $\delta$).
Then it has to be the case that the score of one chain will grow faster
than the other one as shown in Claim~\ref{claim:blockweight} (even if both chain keep
growing).
Now, we argue why it is unlikely that both chain keep growing indefinitely.
The weaker chain grows only if the leader on that chain is either Byzantine
or hasn't heard of the latest blocks on the other chain.
As explained in the proof of Claim~\ref{claim:blockweight} for every altruistic player,
for two chains of roughly the same length,
there's half a chance that they receive the weaker chain first due to the unbiasability
and uniform distribution of the random beacon.
Thus for an altruistic player there's half a chance that
they extend the weaker chain.
\\ \indent
On the other hand, the stronger chain grows even if the elected leader received
the weaker one first (as long as they are not eligible on it, which happens with 
probability $1-1/n$). 
More formally, in the case where altruistic players are leaders on both chain,
the stronger one will be extended with probability $(1-1/n)+1/n\cdot 0.5= 1- 0.5 \cdot 1/n$
and the probability that the weaker chain grows is $0.5$.
Thus it is more likely for an altruistic player to extend the stronger chain.
This explains why the strongest chain grows with higer probability.
\\ \indent
Thus with high probability, there exists a time $\tau_0$
such that altruistic players will stop extending the weaker chain.
\\ \indent
After this time $\tau_0$ it is very unlikely that players
will revert their main chain to another chain.
Indeed an adversary that tries to revert the main chain does not succeed
except with negligible probability.
Let's assume that at time $\tau_0$, the difference between the main chain
and the chain that the adversary is trying to extend is $m$.
To revert the chain, they 
need to create a competitive DAG with at least $m$ references within the fork faster than the main chain 
grows. 
As $m$ gets biggers, the adversary will need to create more blocks privately and this attack becomes
less lilely to succeed as the probability of creating $\ell$ blocks privately decreases with $\ell$ (we will compute this probability in the proof of the next Theorem~\ref{thm:quality}).
Here we assume that the proof-of-delay is secure, i.e. that even an adversary
with enough power will not be able to compute a proof-of-delay faster than expected.
\\ \indent
Finally, as explained in Section~\ref{sec:defns-pos}, we must consider
long-range attacks, where an adversary re-writes the history by bribing old 
participants.  Because we add a decentralized checkpointing, this attack will not succeed.
Let's assume that an adversary has bought old keys from previous participants
and re-wrote the history of the blockchain with those.
When they receive this new chain, altruistic players are already
aware of at least one second witness block (as the reconfiguration
period has to start after a second witness block as explained in Section~\ref{sec:protocol-rule}).
According to the betting rule in Section~\ref{sec:protocol-rule} once altruistic players know of a second witness block, 
they will not bet on a block that does not bet on an associated candidate block.
Thus
altruistic players will never bet on the new adversarial chain.
This is also true for rational players since
when they see this new chain, they would have to
start ignoring all the blocks they have created in order to bet on it
(due to the fourth check in $\verifyblock$),
thus losing most of their deposit.
Thus the chain created with old keys will not be accepted by current
participants.
\\ \indent
We have thus shown that after time $\tau_0$, altruistic
players have agreed on the main chain $C_{k_0}$ and that it's very unlikley they will
revert to another main chain.
This proves the result.
\end{proof}
}
\iffull\forkresistance\fi
Next, we show chain quality.  Intuitively, this results from the fact that a 
player is unlikely to be elected a winner for many consecutive blocks
and that the proof-of-delay is secure. 
\begin{theorem}[Chain quality]\label{thm:quality}
A coalition $\coalfrac<1/3$ of non-altruistic players cannot
contribute to a fraction of more than $\mu=\coalfrac+\alpha$ of the blocks in
the main chain. Given our choice of values for the different parameter we
have $\alpha = 0.03$.
\end{theorem}
\newcommand\quality{
\begin{proof}
The probability that an adversary controlling $\ncoal$ participants 
can contribute one block is $1-(1-p)^{\ncoal}$.
Thus we have that the probability that a player contributes exactly
$\ell$ consecutive blocks, without grinding is
$(1-(1-p)^{\ncoal})^\ell(1-p)^{\ncoal}$.
The expected number of consecutive blocks is thus
$\sum_{j=0}^\infty j(1-(1-p)^{\ncoal})^j(1-p)^{\ncoal}
=(1-p)^{\ncoal}\sum_{j=0}^\infty j(1-(1-p)^{\ncoal})^j
=(1-p)^{\ncoal}\cdot \frac{(1-(1-p)^{\ncoal})}{(1-(1-(1-p)^{\ncoal}))^2}
=(1-p)^{\ncoal}/(1-(1-p)^{\ncoal})$.
(This is the case where a player re-draw the lottery once each time they
are elected leaders.)
A quick estimation shows that for $n$ big enough
the expected number of consecutive blocks
for a coalition of a third that plays honestly is 0.395. 
For a coalition that grinds, the idea is that they will 
``re-draw'' the lottery for each of their
winning shares to try and create more blocks.
The probability of creating a chain of exactly $\ell$ consecutive blocks
when grinding is
$\sum_{x_1,\dots,x_\ell\in S_\ell}\binom{\ncoal}{x_1}\dots\binom{\ncoal x_{\ell-1}}{x_\ell}
p^{\sum x_i}(1-p)^{n(1+\sum_{i=1}^{\ell}x_i)-\sum_{i=1}^{\ell}x_i}$,
where $S_\ell=\{x_1,\dots,x_\ell : 1\le x_1\le \ncoal \ ; \forall i>1 \ 1\le x_i\le x_{i-1}\ncoal \}$.
(This is a sequence of Bernoulli trials where the number of trials is
the number of successes on the previous round times $\ncoal$.)
Using the above probability, a player
that grinds through all their blocks has an expectation of creating 0.42 blocks.
This gives a value of $\alpha=0.03$. We will confirm this value in the simulation in the
next section.
We thus see that the advantage gained by grinding is limited,
due to our leader election mechanism.
Again, we assume that the proof-of-delay is secure.
\end{proof}
}
\iffull\quality\fi
Finally, we show robustness.  The main reason this holds is that, by 
following the protocol in betting on the $\fcr$, a block gets more references 
and thus has a higher score than a block not following the protocol.
A coalition of players can gain a small advantage by grinding through
all the blocks they can create, but when doing so they keep their chain
private and thus prevent other players from referencing it.  This in 
turn reduce the rewards associated with these blocks.

\begin{theorem}[Robustness]\label{thm:robustness}
Given the utility function in Equation~\ref{eq:utility} and the values for 
$\rwd$ and $\pun$ chosen in Section~\ref{sec:sim-settings}, following the 
protocol is a $\epsilon$-$(1/3,1/4)$-robust equilibrium, where $\epsilon =
1.1$.
\end{theorem}
\newcommand\robustness{
\begin{proof}
In order to get an intuition behind the proof, we
first show that following the protocol is a Nash equilibrium.
Recall that the possible choices when elected leader are: (1) whether to bet or not, 
(2) which leaves to include, and (3) when to broadcast their blocks.  
We show that for each of these, if other players follow the protocol then a 
player is incentivized to follow the protocol.
If the player is elected leader on the $\fcr$, we want to show that betting on it gives
them a higher probability of being a winner.  This is because, by definition of the 
$\fcr$ (Algorithm~\ref{alg:fcr}), betting on it means betting on the stronger
chain.  As argued in Claim~\ref{claim:blockweight}, this will add more to
its score, and thus again by the definition of the $\fcr$ it has a higher probability of
being the next block chosen by the $\fcr$. (It could still, however, lose against
another bet on the $\fcr$ that has a smaller hash, but even in this case
the $\lab$ function still labels it as $\neutral$.) This establishes (1).
\\ \indent
By creating a block on top of a weaker chain, a player needs to ignore
the stronger chain, which means referencing fewer blocks (i.e., 
ignoring blocks of which they are aware).  This means that their block will have 
worse connectivity, however, and thus has a higher chance of being
labelled $\loser$ and thus getting a punishment.  This is because, by the 
definition of the anticone, worse connectivity means a bigger anticone, which
in turns means a bigger intersection with the $\blue$ set and thus, by the 
definition of $\lab$ (Algorithm~\ref{alg:label}), a higher chance of being
labelled $\loser$.  This establishes (2). \saraha{could add some proba here
(expected utility is prob of winning * reward + prob of losing*punishment and punishment
is higher than reward.}
\\ \indent
To argue about (3), we now show why a rational player is incentivized to
reveal their block as soon as possible. 
By broadcasting their block as soon as they created it, 
their blocks can get more references (since other player follow the protocol),
which again increases the probability of being a winner, and the expected
accompanying reward. (A single player has nothing to grind through.)
\\ \indent
Now, in order to show that the protocol is robust, we show that a 
coalition of rational players that deviate from the protocol to raise 
their utility, can only do so by $\epsilon$ (resiliency).
We next show that a Byzantine adversary cannot decreases the utility
of honest players by more than $1/ \epsilon$ (immunity).
\\ \indent
\paragraph*{Resiliency}
We consider a coalition of a fraction $\coalfrac$ of 
rational players.
Because it is costless to create blocks, a rational coalition
can clearly gain an advantage by grinding
through all the blocks they can create in order to find a subDAG
that increases their utility. However due to the restrictions imposed by
the leader election and assuming the proof-of-delay is secure
the advantage they can gain is limited.
As shown in Theorem~\ref{thm:quality}, the expected number of blocks
contributed by an adversary that adopts a grinding strategy
is 0.42
versus 0.39 for a coalition that follows the rule. Rational participants can thus increase their gain
from $0.39c$ to $0.42c$.
We thus achieve $\epsilon$-robustness with 
$\epsilon=0.42/0.39=1.077$. This results will be confirmed
by the simulations.

\paragraph*{Immunity}
We need to show that even in the case where a fraction $t$ of players behave 
completely irrationally, the outcome of the rest of the players stays unchanged.
According to the utility functions, as defined in Section~\ref{sec:protocol},
there are three independant components to the utility function that an adversary
could try to influence to harm altruistic players: (1) the $\rwd$ term
(2) the $\pun$ term and (3) the $\bigpun$ term. To harm the honest player,
an adversary could try:
(1) preventing an honest players from contributing blocks to the main chain;
(2)-(3) increasing the number of blocks from the honest
players that gets punished by $\pun$ or $\bigpun$.
The adversary cannot incur any $\bigpun$ to the altruistic players since
they cannot force them to ignore their own block, thus we only focus on case (1)
and (2).
To do (1) an adversary cannot indeed prevent players from creating blocks but once
they produce it, they can try and create an alternative blockDAG
so that the altruistic player's block does not make it to the main chain.
To do (2), the adversary could create an alternative blockDAG
that does not reference altruistic players' blocks
to 
try and incur a punishment to their blocks.
(According to the $\lab$ function defined in Section~\ref{sec:protocol}
a block gets a punishment if its anticone intersects the blue set
for more than $k$ blocks and a block that has less connections to other blocks
has a bigger anticone.)
In both cases the Byzantine adversary harms a player the most when creating
the biggest alternative blockDAG that does not reference altruistic players'
blocks.
\input{graphs}
To incur a punishment to the altruistic players
Byzantine players need to create a subDAG of more than $k$ blocks on their own,
where $k$ is the interconnectivity paramater.
Indeed if they do so then those $k$ blocks will be
in the anticone of an altruistic player that contributed a block
at that same time and will be labeled a loser according to
the $\lab$ algorithm in Section~\ref{sec:protocol-rule}.
When choosing $k=3$ and using similar probabilities as in Theorem~\ref{thm:quality},
one can compute that
the probability of creating a subDAG of more than 3 blocks
is
$(1-p)^{\ncoal}+ 
{\ncoal}*p*(1-p)^{2\ncoal-1}+
\binom{n}{2}p^2(1-p)^{3\ncoal-2}
+(np(1-p)^{\ncoal-1})^2 (1-p)^{\ncoal}+
\binom{\ncoal}{3}p^3(1-p)^{4\ncoal-3}+
\binom{\ncoal}{2}\ncoal p^3 (1-p)^{3\ncoal-2}+
\binom{\ncoal}{2}2\ncoal p^3 (1-p)^{4\ncoal-3}+
({\ncoal}p(1-p)^{(\ncoal-1)})^{3}(1-p)^{\ncoal}
$.
An estimation of the previous probability gives
2.5\% for an adversary that controls a third of the player.
\\ \indent
This means that 2.5\% of the blocks the adversary create can incur a punishment of $\pun$ to another player. Since on
average two third of the players should contribute to two third
of the blocks, we have that a third of Byzantine players will reduce the
payoff of the other players from $67\cdot c$ to $67\cdot c-2.5\cdot\pun$ every hundred blocks.
With our numerical value of $\pun=6$ and $c=1$ we have that the payoff of the other players is reduced
by 67/52=1.29.
For a coalition of a quarter the above probability is 1\%, bringing the above ratio
to 67/61=1.1. With $\epsilon=1.1$,
we conclude that a coalition of one quarter of Byzantine
player cannot harm the others since they cannot create an alternative blockDAG with
enough advantage. 
The protocol is thus $\epsilon-$immune against a coalition of a quarter. 
This will also be confirmed by our simulations Section~\ref{sec:simulations}.
\end{proof}
}
\iffull
\begin{proof}
In order to get an intuition behind the proof, we
first show that following the protocol is a Nash equilibrium.
Recall that the possible choices when elected leader are: (1) whether to bet or not, 
(2) which leaves to include, and (3) when to broadcast their blocks.  
We show that for each of these, if other players follow the protocol then a 
player is incentivized to follow the protocol.
If the player is elected leader on the $\fcr$, we want to show that betting on it gives
them a higher probability of being a winner.  This is because, by definition of the 
$\fcr$ (Algorithm~\ref{alg:fcr}), betting on it means betting on the stronger
chain.  As argued in Claim~\ref{claim:blockweight}, this will add more to
its score, and thus again by the definition of the $\fcr$ it has a higher probability of
being the next block chosen by the $\fcr$. (It could still, however, lose against
another bet on the $\fcr$ that has a smaller hash, but even in this case
the $\lab$ function still labels it as $\neutral$.) This establishes (1).
\\ \indent
By creating a block on top of a weaker chain, a player needs to ignore
the stronger chain, which means referencing fewer blocks (i.e., 
ignoring blocks of which they are aware).  This means that their block will have 
worse connectivity, however, and thus has a higher chance of being
labelled $\loser$ and thus getting a punishment.  This is because, by the 
definition of the anticone, worse connectivity means a bigger anticone, which
in turns means a bigger intersection with the $\blue$ set and thus, by the 
definition of $\lab$ (Algorithm~\ref{alg:label}), a higher chance of being
labelled $\loser$.  This establishes (2). \saraha{could add some proba here
(expected utility is prob of winning * reward + prob of losing*punishment and punishment
is higher than reward.}
\\ \indent
To argue about (3), we now show why a rational player is incentivized to
reveal their block as soon as possible. 
By broadcasting their block as soon as they created it, 
their blocks can get more references (since other player follow the protocol),
which again increases the probability of being a winner, and the expected
accompanying reward. (A single player has nothing to grind through.)
\\ \indent
Now, in order to show that the protocol is robust, we show that a 
coalition of rational players that deviate from the protocol to raise 
their utility, can only do so by $\epsilon$ (resiliency).
We next show that a Byzantine adversary cannot decreases the utility
of honest players by more than $1/ \epsilon$ (immunity).
\\ \indent
\paragraph*{Resiliency}
We consider a coalition of a fraction $\coalfrac$ of 
rational players.
Because it is costless to create blocks, a rational coalition
can clearly gain an advantage by grinding
through all the blocks they can create in order to find a subDAG
that increases their utility. However due to the restrictions imposed by
the leader election and assuming the proof-of-delay is secure
the advantage they can gain is limited.
As shown in Theorem~\ref{thm:quality}, the expected number of blocks
contributed by an adversary that adopts a grinding strategy
is 0.42
versus 0.39 for a coalition that follows the rule. Rational participants can thus increase their gain
from $0.39c$ to $0.42c$.
We thus achieve $\epsilon$-robustness with 
$\epsilon=0.42/0.39=1.077$. This results will be confirmed
by the simulations.

\paragraph*{Immunity}
We need to show that even in the case where a fraction $t$ of players behave 
completely irrationally, the outcome of the rest of the players stays unchanged.
According to the utility functions, as defined in Section~\ref{sec:protocol},
there are three independant components to the utility function that an adversary
could try to influence to harm altruistic players: (1) the $\rwd$ term
(2) the $\pun$ term and (3) the $\bigpun$ term. To harm the honest player,
an adversary could try:
(1) preventing an honest players from contributing blocks to the main chain;
(2)-(3) increasing the number of blocks from the honest
players that gets punished by $\pun$ or $\bigpun$.
The adversary cannot incur any $\bigpun$ to the altruistic players since
they cannot force them to ignore their own block, thus we only focus on case (1)
and (2).
To do (1) an adversary cannot indeed prevent players from creating blocks but once
they produce it, they can try and create an alternative blockDAG
so that the altruistic player's block does not make it to the main chain.
To do (2), the adversary could create an alternative blockDAG
that does not reference altruistic players' blocks
to 
try and incur a punishment to their blocks.
(According to the $\lab$ function defined in Section~\ref{sec:protocol}
a block gets a punishment if its anticone intersects the blue set
for more than $k$ blocks and a block that has less connections to other blocks
has a bigger anticone.)
In both cases the Byzantine adversary harms a player the most when creating
the biggest alternative blockDAG that does not reference altruistic players'
blocks.
\input{graphs}
To incur a punishment to the altruistic players
Byzantine players need to create a subDAG of more than $k$ blocks on their own,
where $k$ is the interconnectivity paramater.
Indeed if they do so then those $k$ blocks will be
in the anticone of an altruistic player that contributed a block
at that same time and will be labeled a loser according to
the $\lab$ algorithm in Section~\ref{sec:protocol-rule}.
When choosing $k=3$ and using similar probabilities as in Theorem~\ref{thm:quality},
one can compute that
the probability of creating a subDAG of more than 3 blocks
is
$(1-p)^{\ncoal}+ 
{\ncoal}*p*(1-p)^{2\ncoal-1}+
\binom{n}{2}p^2(1-p)^{3\ncoal-2}
+(np(1-p)^{\ncoal-1})^2 (1-p)^{\ncoal}+
\binom{\ncoal}{3}p^3(1-p)^{4\ncoal-3}+
\binom{\ncoal}{2}\ncoal p^3 (1-p)^{3\ncoal-2}+
\binom{\ncoal}{2}2\ncoal p^3 (1-p)^{4\ncoal-3}+
({\ncoal}p(1-p)^{(\ncoal-1)})^{3}(1-p)^{\ncoal}
$.
An estimation of the previous probability gives
2.5\% for an adversary that controls a third of the player.
\\ \indent
This means that 2.5\% of the blocks the adversary create can incur a punishment of $\pun$ to another player. Since on
average two third of the players should contribute to two third
of the blocks, we have that a third of Byzantine players will reduce the
payoff of the other players from $67\cdot c$ to $67\cdot c-2.5\cdot\pun$ every hundred blocks.
With our numerical value of $\pun=6$ and $c=1$ we have that the payoff of the other players is reduced
by 67/52=1.29.
For a coalition of a quarter the above probability is 1\%, bringing the above ratio
to 67/61=1.1. With $\epsilon=1.1$,
we conclude that a coalition of one quarter of Byzantine
player cannot harm the others since they cannot create an alternative blockDAG with
enough advantage. 
The protocol is thus $\epsilon-$immune against a coalition of a quarter. 
This will also be confirmed by our simulations Section~\ref{sec:simulations}.
\end{proof}
\fi

\subsection{Simulations}\label{sec:simulations}

We next present our simulations, that support and validate
the proofs above.
As stated in Section~\ref{sec:model}, we consider three types of players: 
Byzantine, altruistic, and rational.
We simulate the game using different fractions of different types of players.
Our simulation is written in Python, and consists of roughly 1,000 lines of 
code.  All players start with the same deposit.
To model network latency, we add random delays between the propagation of
blocks amongst players.  Following Decker and
Wattenhofer~\cite{information-propagation}, this random delay follow an
exponential distribution.
Each simulation that we run has 150 players and
lasts for 5,000 time slots.  All our results are
averaged over $120$ runs of the simulation.
\\ \indent
In addition to the balance of the types of players, there are several 
different parameters we need to consider: $k$, the inter-connectivity 
parameter; $c$, the constant in the reward in Equation~\ref{eq:utility}; 
$\mathsf{pun}$, the punishment; and $\bigpun$, the big punishment.  We 
chose $k = 3$, $c = 1$, $\pun = 6$, and $\bigpun = 10$.  We also 
define the initial deposit of all players to be $0$, but allow for payoffs to 
go below zero.  We stress 
that all these values are relatively arbitrary, as we are more interested in
the ratio between them and their evolution rather than their specific values.  
Different parameters do, however, result in different effects on the protocol.  
For example, decreasing the value of the punishment would increase the
immunity of the protocol, but would also weaken its resiliency (since a
coalition would be able to gain a bigger profit).  We leave a more in-depth
exploration of this trade-off for future work.
\\ \indent
Because we do not use our simulation to support our argument for the growth of the chain,
we do not implement the proof-of-delay function discussed at the end of
Section~\ref{sec:caucus-construction} (as it is crucial only for liveness).
\paragraph*{Stategies}
As explained in Section~\ref{sec:sim-settings}, the strategy
followed by rational players is to create all the blocks they can
privately and then grind through all the subDAG to find the one that increases their
expected utility and broadcast this subDAG only.
For the Byzantine players, we are interested in capturing the
\emph{worst} type of adversary in terms of (1) decreasing
the payoff of altruistic players (this will capture the immunity
property) and (2) the biggest fork they
can create (to capture the convergence property).
We explained in the proof of Theorem~\ref{thm:robustness} that the biggest harm Byzantine
players can incur to altruistic players is by creating the biggest subDAG
possible.
Thus we consider Byzantine adversaries that create as many blocks as possible
and broadcast them.
\paragraph*{Results}
We measured the length of the longest fork; the results
are in Figure~\ref{fig:fork}.  For a maximum number of 49 Byzantine players,
the longest fork is 12 blocks, so the simulation supports
our proof that the protocol converges.
\\ \indent
In 
Figure~\ref{fig:quality}
we see that the maximum fraction of blocks contributed by a fraction of up to
one-third of non-altruistic players is $0.351$ compared to an expected 
contribution of $49/150=0.327$.  This shows the chain quality property 
with $\alpha = 0.024$, meaning that grinding through all the 
blocks they are eligible to produce allows an adversary to contribute $2.4\%$
more blocks on average. This is consistent with our proof.
\\ \indent
For robustness, the results of the simulations are in Figures~\ref{fig:rat-payoff}
and~\ref{fig:immunity1}.
We see that a coalition of up to one-third of rational players
increases their payoff up to 3.8 compared to 3.5
when not forming a coalition. We thus achieve $\epsilon$-robustness with 
$\epsilon=1.086$.
To quantify the harm that a Byzantine coalition can do to others,
we compute for each simulation the payoff of altruistic players in the 
case of a fraction $t$ of Byzantine players trying to harm the players.
Figure~\ref{fig:immunity1} shows the payoffs for both altruistic and Byzantine 
players.  We see that the payoffs of altruistic players is unaffected even in
the presence of a quarter of Byzantine players, but starts decreasing after
the coalition becomes larger.  This is why we claim robustness only for 
coalitions up to a quarter, but do observe here that the payoff of the
Byzantine coalition is much more negatively affected than that of the 
altruistic players. These results are again consistent with our previous proof.

%% file: graphs.tex
\begin{figure*}[th]
\centering
\begin{subfigure}[b]{0.48\linewidth}
\centering
\includegraphics[width=0.8\linewidth]{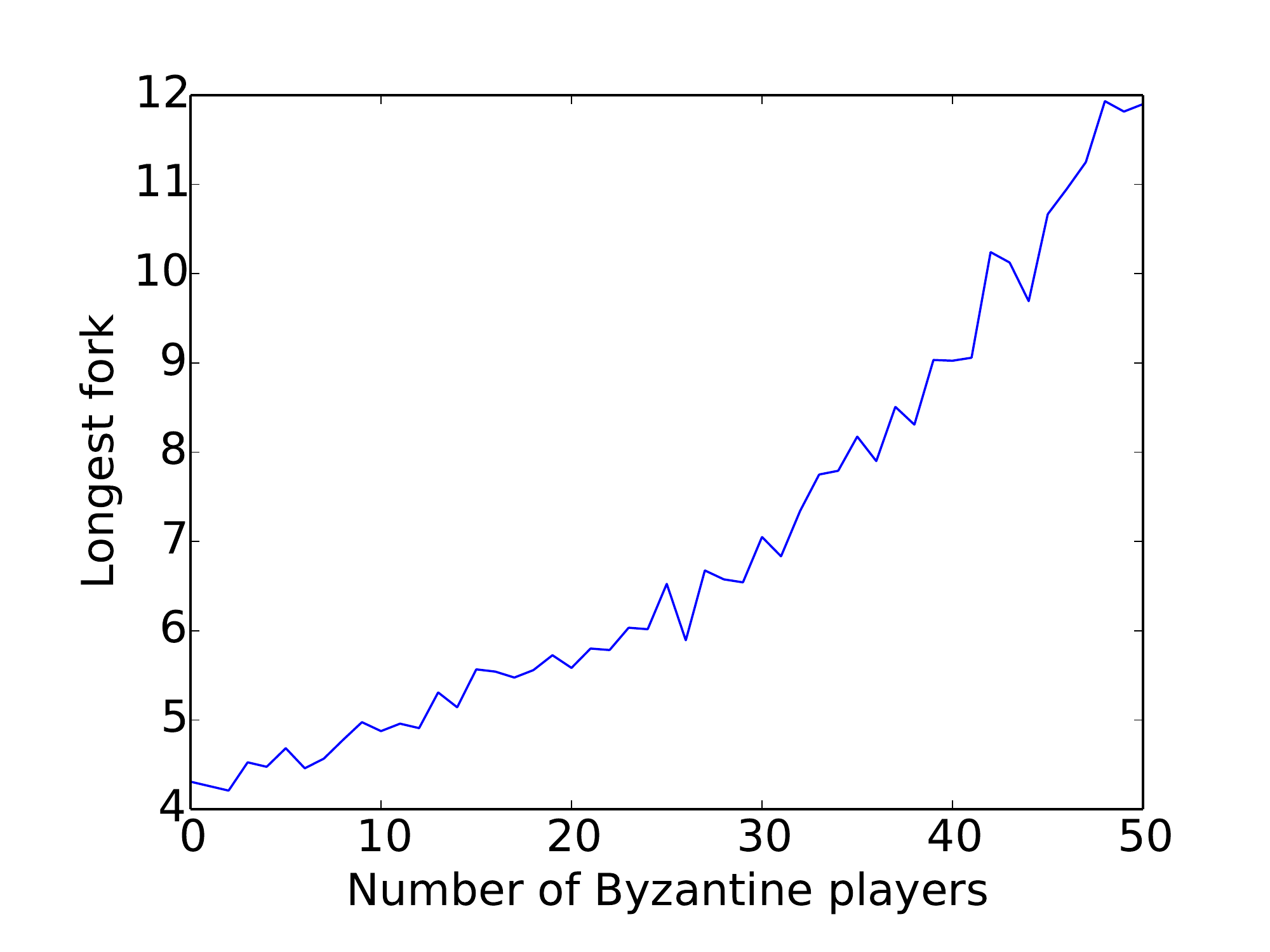}
\caption{The length of the longest fork, in the presence of a coalition of 
Byzantine players.}
\label{fig:fork}
\end{subfigure}
~\begin{subfigure}[b]{0.48\linewidth}
\centering
\includegraphics[width=0.8\linewidth]{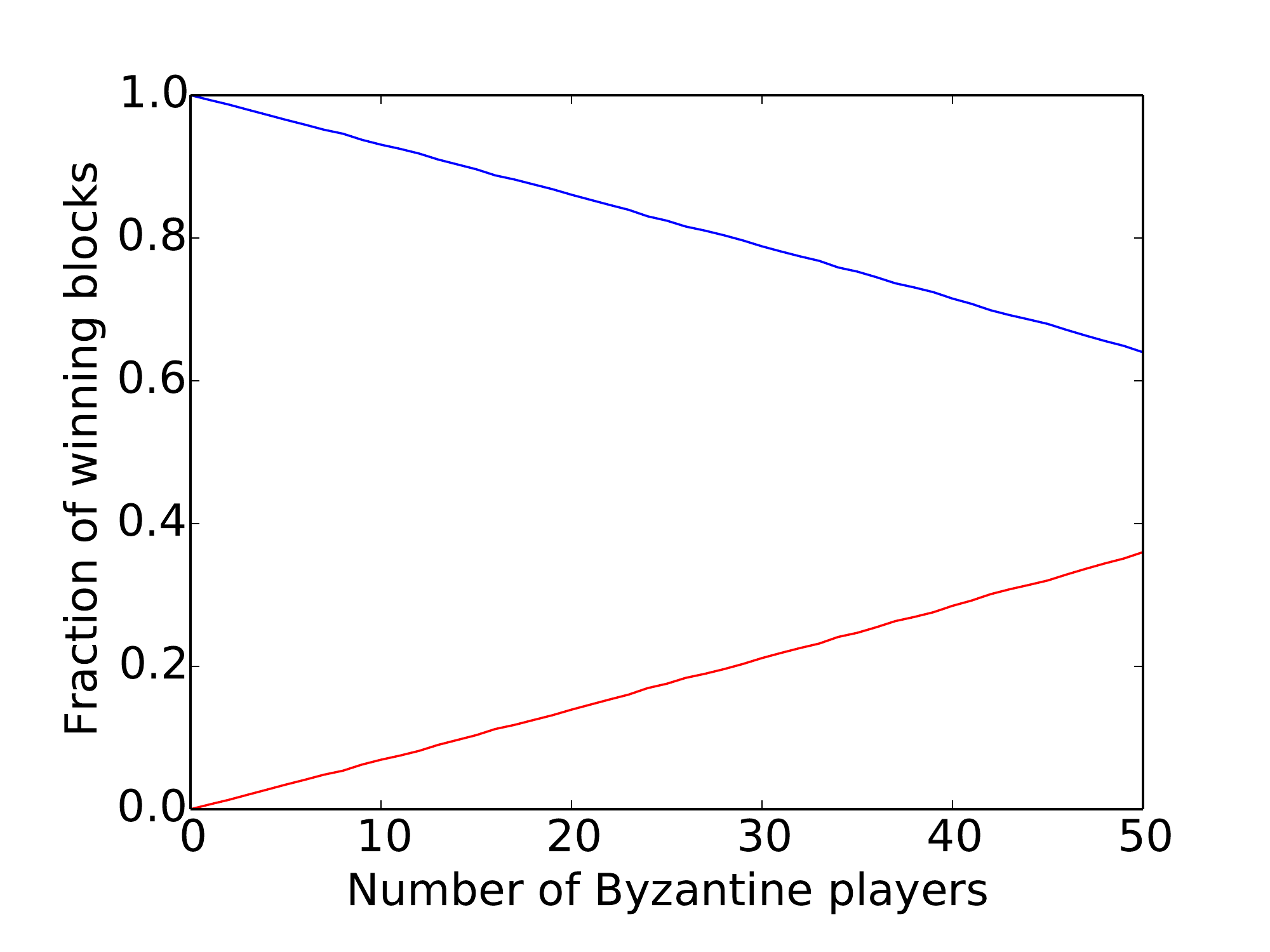}
\caption{The fraction of winning blocks belonging to altruistic (blue) and 
Byzantine (red) players.}
\label{fig:quality}
\end{subfigure}
~\\
\begin{subfigure}[b]{0.48\linewidth}
\centering
\includegraphics[width=0.8\linewidth]{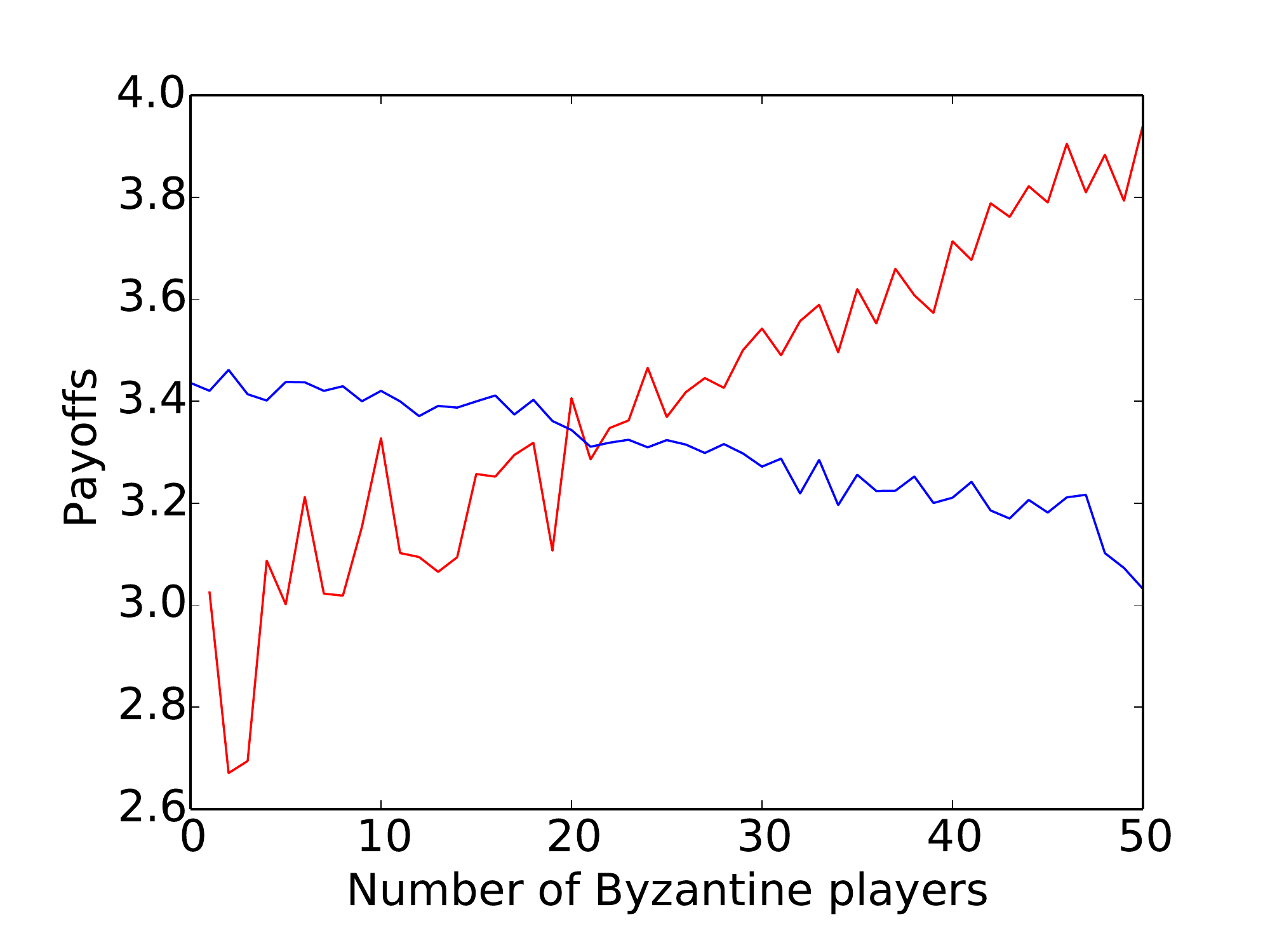}
\caption{The payoff for altruistic players (blue) and a coalition of
rational players (red).}
\label{fig:rat-payoff}
\end{subfigure}
~
\begin{subfigure}[b]{0.48\linewidth}
\centering
\includegraphics[width=0.8\linewidth]{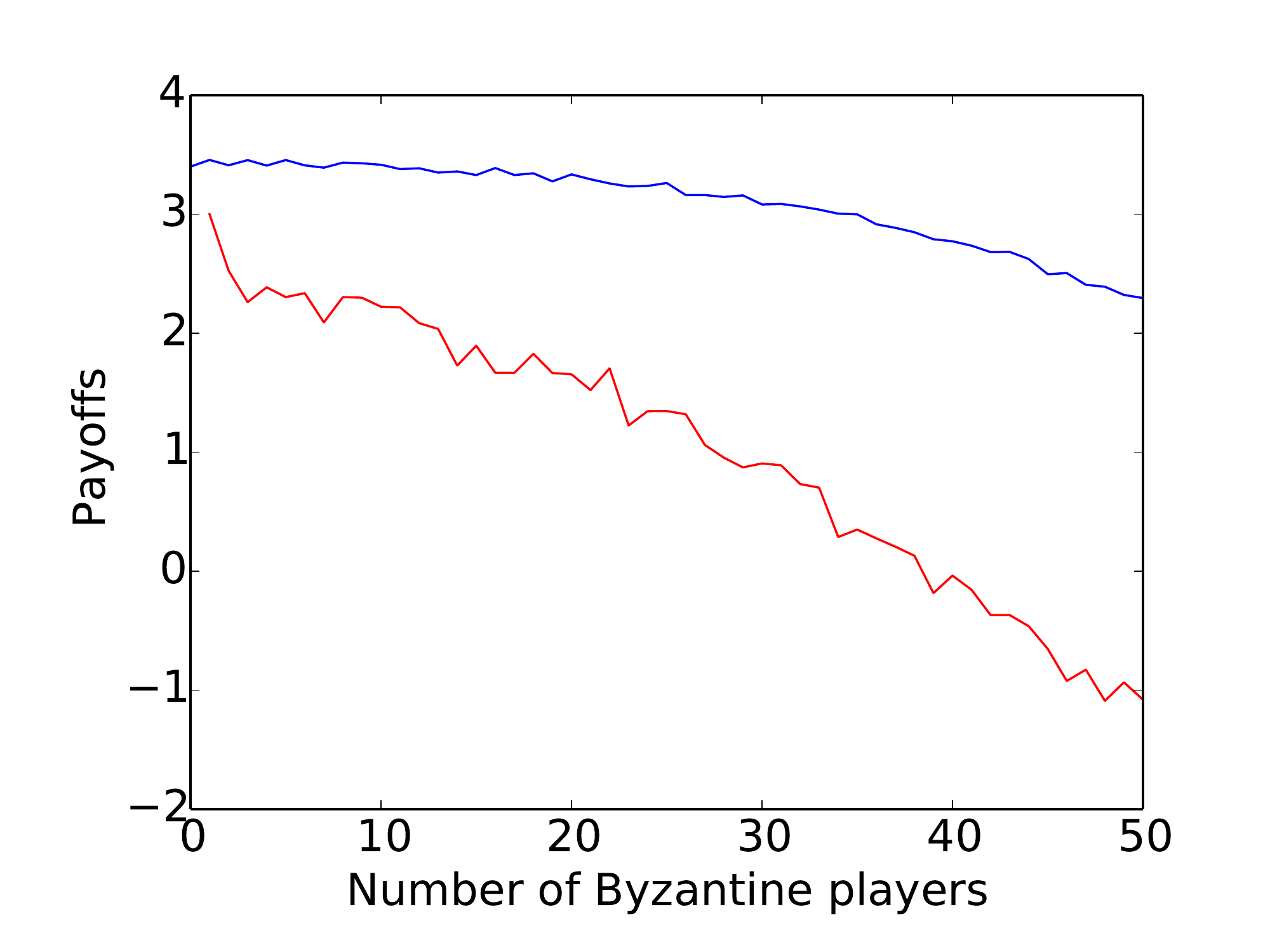}
\caption{The payoff for altruistic players (blue) in the presence of a
coalition of Byzantine players (red).}
\label{fig:immunity1}
\end{subfigure}
\caption{Results from our simulation, averaged over 10 runs and considering
up to 50 non-altruistic players (out of a total of 150).}
\label{fig:sims}
\end{figure*}

%% file: conclusion.tex
\section{Conclusions and Future Work}

This paper presented \betting, a blockchain-based consensus protocol composed 
of two components of potential independent interest: a leader election protocol,
\sysname, and a scheme for incentivization.
The protocol is secure in a semi-synchronous setting against a coalition
consisting of up to a third of participants, and achieves 
finality via the use of decentralized checkpointing.
\\ \indent
While \betting makes some important first steps in treating incentives
as a first-class concern, there are other avenues to consider.
%
In terms of evaluation, existing literature analyzing incentives in PoW-based 
systems has used techniques like Markov Decision Processes~\cite{Gervais:2016} 
or no-regret learning~\cite{bitcoin-without-block-reward} in order to justify 
more formally the best rational strategy.  These techniques would be much more
difficult to apply in a setting with PoS, but it would nevertheless be useful
to better justify the Byzantine strategy used in our simulations.
\ifndss
\\ \indent
Finally, we currently treat all bets as being of equal value (one
block is one bet), but it may be interesting to consider bets of variable 
size, in which players that are more confident about the blocks on which they
place bets (for example, because those blocks are highly connected) could
attempt to gain a higher reward by placing a bet of higher value.
\saraha{+ transistion from pow to pos}\fi